\documentclass[11pt]{article}
\usepackage{rotating}
\usepackage{natbib}
\usepackage{amsfonts, amsmath,amssymb,bm, graphicx,color}
\usepackage{enumerate}
\usepackage{moreverb,booktabs,longtable}
\usepackage{amsmath,amssymb,bm, graphics,epsfig,color}
\usepackage{subfigure}
\usepackage{amsthm,bbm}
\usepackage{moreverb,booktabs,multirow,lscape,url}
\usepackage{bbm}
\usepackage{epstopdf}
\usepackage{pdfpages}
\usepackage[dvipsnames]{xcolor}
\usepackage{siunitx}
\usepackage{threeparttable}  
\usepackage{setspace}
\usepackage{arydshln}
\usepackage{comment}
\usepackage[linesnumbered,ruled,vlined]{algorithm2e}

\makeatletter
\renewenvironment{abstract}
{\small
	\begin{center}%
		{\bfseries \abstractname\vspace{-.5em}\vspace{\z@}}%
	\end{center}%
	\list{}{\leftmargin=0pt\rightmargin=0pt}%
	\item\relax}
{\endlist}
\makeatother

\usepackage[colorlinks=true,urlcolor=cyan,citecolor=blue,linkcolor=blue,bookmarks=true]{hyperref}
\newtheorem{The}{Theorem}

\numberwithin{The}{section}
\allowdisplaybreaks

\def\qed{\hfill \vrule height 6pt width 6pt depth 0pt}

\usepackage{footnote}
\usepackage[affil-it]{authblk}
\usepackage[english]{babel}

\begin{document}

\title{\bf Multi-Trigger Crypto CAT Bonds with On-Chain Settlement: Valuation and Optimal Design}
 \author{Yue Wang$^{1,\dagger}$, ~ Yijia Li$^{2,\dagger}$,  ~ Maochao Xu$^{3}$\thanks{Corresponding author: mxu2@ilstu.edu. $\dagger$ Both authors contributed equally to this work. }, ~ Xianyue Li$^{1}$\\
     $^{1}$    School of Mathematics and Statistics, Lanzhou University, China\\
    $^{2}$Department of Statistics and Finance, University of Science and Technology of China, China\\
        $^{3}$Department of Mathematics, Illinois State University, US
 
    }

\maketitle

\begin{abstract}
	
	\begin{itemize}
				
		\item \textbf{Purpose:} This study develops a pricing and contract design framework for cryptocurrency catastrophe (CAT) bonds to transfer extreme crypto-native risks, including protocol exploits, exchange breaches, and decentralized finance (DeFi) failures, to capital markets. The paper aims to address arbitrage-free valuation, sponsor-optimal contract design, and trustless settlement under the unique informational and operational features of blockchain systems.
		
		\item \textbf{Design/methodology/approach:} We propose a multi-trigger crypto CAT bond structure that jointly captures short-term catastrophic shocks and long-term systemic deterioration through oracle-reported loss metrics. An arbitrage-free valuation framework is developed under an incomplete-market setting using the minimal martingale measure, while sponsor-optimal contract design is formulated under a dual-measure framework. Empirically, crypto loss dynamics are modeled using generalized extreme value distributions and copula-based dependence structures, whereas financial risk factors are modeled through ARIMA--GARCH and vine copulas. A smart-contract-enabled on-chain settlement architecture is further introduced to automate trigger evaluation and cash-flow execution.
		
		\item \textbf{Findings:}  Empirical results based on REKT crypto incident data demonstrate strong dependence between monthly extreme and aggregate losses, with heterogeneous dependence structures across blockchain ecosystems. Simulation studies show that trigger and principal repayment designs substantially affect bond price distributions and tail risk exposures. Conservative trigger structures generate more stable bond valuations, whereas aggressive structures exhibit greater downside dispersion. The proposed framework supports economically viable risk transfer while enabling transparent and timely settlement through blockchain-based execution.
		
		\item \textbf{Originality/value:} This study develops, to the best of our knowledge, the first integrated framework for crypto-native catastrophe bonds that combines arbitrage-free pricing, sponsor-optimal contract design, and smart-contract-based on-chain settlement. Unlike traditional CAT bonds or cyber insurance-linked securities, the proposed framework explicitly incorporates oracle-based observability, crypto-specific dependence structures, and automated settlement, providing a novel mechanism for transferring systemic digital asset risks to capital markets.
		
	\end{itemize}
	
\end{abstract}
{\bf keywords:} {Crypto risks; Copula modeling; GARCH models; Smart contract}

\newpage
\section{Introduction and Motivation}
Cryptocurrency (commonly referred to as crypto) denotes digital assets secured by cryptographic techniques and maintained on decentralized, distributed ledger systems. These assets bypass traditional intermediaries and include a wide range of tokens with varied utilities. As of June 2025, the total crypto market capitalization exceeds USD 3.3 trillion, led by Bitcoin (USD 2.1 trillion) and Ethereum (USD 293 billion)\footnote{\url{https://coinmarketcap.com/}}. The decentralized finance (DeFi) sector, built atop programmable blockchains like  Ethereum (ETH), Binance Smart
Chain (BSC), supports applications in lending, trading, and yield farming. The total value locked (TVL) across DeFi protocols stands at approximately USD 241 billion\footnote{\url{https://defillama.com/}}. 

Despite its technological promise, the crypto ecosystem remains highly susceptible to cyber risks, including phishing attacks, compromised access controls, flash loan exploits, and rug pulls. For instance,  On October 6, 2022, the BNB Smart Chain's native cross-chain bridge between the BNB Beacon Chain and the BNB Smart Chain was exploited, leading to the unauthorized minting of 2,000,000 BNB (valued at approximately \$586 million) directly to the attacker’s address. Similarly, on February 21, 2025, Bybit, one of the largest crypto exchanges, suffered a major security breach resulting in the loss of approximately \$1.46 billion in ETH. The attack exploited a multisignature cold wallet vulnerability through a sophisticated combination of phishing tactics and smart contract manipulation. According to the REKT database\footnote{\url{https://de.fi/rekt-database}}, cumulative losses in the crypto space amounted to \$89,203,707,324 as of June 2025. These events underscore the urgent need for robust, scalable insurance solutions to mitigate such extreme  risks in the crypto ecosystem. It is worth noting that while several insurance products are currently available in the market for crypto-related risks, their coverage is often limited and excludes many significant threats, such as smart contract vulnerabilities, protocol failures, and systemic platform risks \citep{li2024framework}.

Catastrophe bonds (CAT) are insurance-linked securities (ILS) traditionally used to transfer low-probability, high-impact risks, such as natural disasters, from insurers to capital markets. Investors in CAT bonds receive coupon payments under normal conditions; however, if a predefined trigger event occurs, the principal is partially or fully forfeited and used to cover insured losses. The literature on CAT bonds is extensive, particularly regarding their design and pricing for traditional perils such as earthquakes, extreme mortality, wildfires, and hurricanes \citep{shao2015catastrophe,  li2023pricing,braun2023cyber,  huang2024storm, li2024mitigating}. A comprehensive overview of recent developments in CAT bond structures can be found in \cite{barrieu2024catastrophe}. Of particular relevance to our work is the emerging topic of cyber CAT bonds. \cite{xu2021data} proposed a multi-period CAT bond structure for data breach events, demonstrating that such instruments can serve as both effective risk transfer mechanisms and attractive investment products. \cite{liu2021extreme} employed extreme value theory (EVT) to design CAT bonds tailored to tail cybersecurity risks in power systems, showing that such instruments could significantly reduce the insolvency risk faced by insurers underwriting cyber insurance for critical infrastructure. \cite{kolesnikov2022cyber} developed a general framework for cyber bonds by fitting frequency–severity loss distributions to public cyber incident data and simulating bond prices, yields, and coupon-setting approaches. More recently, \cite{braun2023cyber} analyzed the feasibility of transferring cyber risks to capital markets through ILS. While these studies have advanced the understanding of cyber risk securitization, none address the distinctive features of crypto ecosystems.

Given the growing severity and frequency of extreme losses in the
cryptocurrency ecosystem, there is an increasing need for risk-transfer
instruments tailored to crypto-native catastrophe risk. In this paper, we
propose a  CAT bond structure designed specifically for the
unique loss dynamics of decentralized financial systems. Unlike traditional
CAT bonds, which primarily transfer natural disaster risk, or emerging cyber
CAT bonds that focus on enterprise data breaches, our framework targets
protocol exploits, cross-chain bridge failures, oracle manipulation, and
system-wide events intrinsic to decentralized finance
\citep{Kwock2022,li2024framework}. A crypto CAT bond leverages distinctive features of blockchain-based systems
along three dimensions. First, crypto losses are predominantly human-driven,
arising from code-level vulnerabilities, governance failures, and design
choices embedded in man-made protocols. Because many decentralized
applications rely on shared code bases or interoperable infrastructure, a
single exploit can rapidly propagate across platforms and generate
chain-wide losses within a short time horizon, a mechanism fundamentally
different from the natural perils underlying conventional ILS \citep{zhou2023sok}. Second, public blockchains generate
high-frequency, transparent, and tamper-resistant data, enabling the
construction of parametric triggers based on observable on-chain loss
metrics. This contrasts with traditional CAT bonds, which typically rely on
third-party loss indices or delayed post-event assessments. Third, while the
bond may be issued through a conventional bankruptcy-remote special purpose
vehicle (SPV), trigger evaluation, payout determination, and cash-flow
settlement can be executed on-chain via smart contracts. This hybrid architecture combines the
legal robustness and investor familiarity of traditional ILS issuance with
the automation, transparency, and trust-minimized execution of blockchain
systems \citep{gan2023decentralized,john2023smart}. Taken together, these
features render the crypto CAT bond not merely a digital analogue of existing
insurance-linked securities, but a distinct risk-transfer instrument whose
operational logic is tightly coupled to the informational structure of
decentralized networks \citep{biais2023advances}.

Motivated by this hybrid design, we focus on crypto CAT bonds that are issued
off-chain through a traditional SPV but whose trigger monitoring and
settlement are implemented on-chain. In such settings, loss realizations enter the contract exclusively through
decentralized oracles\footnote{In blockchain systems, decentralized oracles are
third-party mechanisms that transmit off-chain or aggregated information such
as price, event, or loss data to smart contracts, enabling on-chain contracts to
condition payoffs on external states.} rather than through direct verification
of the underlying economic loss \citep{john2023smart,biais2023advances}.
 Consequently,
bond payoffs depend on oracle-reported states. These features are modeled as intrinsic components of
the contract’s information environment and are explicitly incorporated into
the valuation and hedging framework developed in this paper. To the best of our knowledge, \emph{this study develops the first integrated pricing
and contract design framework for crypto-native catastrophe bonds that combine
traditional SPV-based issuance with on-chain trigger adjudication and
settlement, while explicitly incorporating oracle-based observability.}

{

These features raise fundamental questions regarding arbitrage-free valuation,
risk decomposition, and contract design for crypto CAT bonds operating under
oracle-based observability and market incompleteness. This paper addresses the
following interrelated research questions.
\begin{itemize}
    \item RQ1: {\em Arbitrage-free valuation under on-chain constraints.}
Does an arbitrage-free pricing framework exist for crypto CAT bonds whose
payoffs depend on non-tradable crypto-loss dynamics and on-chain execution
states? In particular, can one construct an equivalent martingale measure on
an enlarged filtration that incorporates both financial market information and
oracle-observed loss processes?

\item RQ2: {\em Pricing and hedging in incomplete crypto markets.}
In an inherently incomplete crypto market, which pricing measure provides a
principled benchmark for valuation and risk decomposition? How does the minimal
martingale measure associated with the tradable financial submarket separate
hedgeable financial risk from irreducible crypto-loss, and what
are the implications for hedging and residual risk?

\item RQ3: {\em Sponsor-optimal contract design under dual measures.}
How should a crypto CAT bond be designed when market valuation is governed by a
martingale pricing measure while sponsor losses are assessed under the physical
measure? What can be established about existence and implementability of
sponsor-optimal coupon and principal schedules under smart-contract constraints,
and how do candidate designs compare in sponsor tail risk?

\item RQ4: {\em Statistical calibration and simulation-based evaluation.}
How can the proposed framework be operationalized with data-driven models for
crypto-loss dynamics and financial discounting factors? Which marginal models
and dependence structures provide a reasonable joint fit, and how do the
calibrated inputs translate into simulated price/return distributions and
sponsor tail-risk metrics for candidate smart-contract designs?
\end{itemize}
The answers to these questions together deliver a unified framework for the
arbitrage-free valuation, contract structuring, and data-driven evaluation of
crypto CAT bonds under oracle-based observability and market incompleteness.

 The rest of the paper is organized as follows.
Section~\ref{sec:framework} introduces the proposed multi-trigger crypto CAT
bond, including the trigger design, cash-flow specification, and oracle-based
on-chain settlement architecture. Section~\ref{sec:na_chain} develops the
theoretical framework for arbitrage-free valuation in an incomplete market and
sponsor-optimal design. Section \ref{sec:statistics1} discusses statistical modeling methodologies: 
Section~\ref{sec:eda} presents exploratory data analysis, while
Sections~\ref{sec:estimation}--\ref{sec:statistics} describe the estimation and
statistical calibration of the joint crypto-loss and financial-factor models.
Section~\ref{sec:empirical} reports simulation-based pricing and design
implications, including distributions of bond prices/returns and sponsor tail
risk across candidate schedules. Section~\ref{sec:case-study} provides a
chain-wide ETH case study to illustrate feasibility. Finally,
Section~\ref{sec:conclusion} concludes; the \emph{Supplementary Material}
contains additional results.

\section{Framework for Crypto CAT Bonds and Pricing}
\label{sec:framework}

In this section, we introduce a market and information framework for the
pricing and design of multi-period crypto catastrophe bonds. The framework
distinguishes between (i) a tradable financial submarket, (ii) non-tradable
crypto-loss dynamics, and (iii) oracle-based observability that governs
on-chain trigger monitoring and settlement. Trading occurs over discrete
time intervals, and the bond is issued at time $t=0$ with maturity $T$.
Coupon payments and trigger evaluation are conducted at predetermined
monitoring dates.
{Within this framework, a crypto CAT bond is characterized by four components:
(i) a trigger mechanism based on oracle-reported loss states, (ii) a
multi-period cash flow structure for coupons and principal, (iii) an
arbitrage-free valuation methodology under an equivalent martingale measure,
and (iv) a smart-contract architecture that enables automated on-chain
execution of trigger adjudication and payments.
}

\paragraph{Market and Information Structure}
\label{subsec:market_structure}

We consider a discrete-time economy on a filtered probability space
$(\Omega,\mathcal{F},\mathbb{P},\mathbb{F}^{\mathrm{chain}})$, where
$\mathbb{F}^{\mathrm{chain}}=(\mathcal{F}^{\mathrm{chain}}_t)_{t=0,\ldots,T}$
represents the public information available to market participants over the bond’s lifetime.
Time is indexed by $t=0,1,\ldots,T$, corresponding to predetermined monitoring and payment dates.

The market comprises two components. The first is a traded financial submarket,
including standard instruments such as money market accounts and inflation- or interest-rate--linked securities.
Traded asset prices are assumed to be adapted to a subfiltration
$\mathbb{F}^{\mathrm{fin}}=(\mathcal{F}^{\mathrm{fin}}_t)_{t=0,\ldots,T}$ and are tradable without frictions.
The second component consists of crypto-related loss processes representing economic losses arising from protocol exploits,
governance failures, and other crypto-native events. These loss processes are non-tradable and therefore cannot be replicated
using the traded financial assets.

Loss realizations are not directly observed by the bond contract. Instead, contract-relevant loss information is provided through
decentralized oracle mechanisms that report aggregated or verified loss metrics on-chain. Let
$\mathbb{F}^{\mathrm{orc}}=(\mathcal{F}^{\mathrm{orc}}_t)_{t=0,\ldots,T}$
denote the filtration generated by oracle reports up to time $t$.
We define the full public information flow by the enlarged filtration
\(
\mathbb{F}^{\mathrm{chain}} = \mathbb{F}^{\mathrm{fin}} \vee \mathbb{F}^{\mathrm{orc}} .
\) Trading is restricted to the traded financial submarket. Accordingly, trading strategies
for the traded assets are required to be predictable with respect to $\mathbb{F}^{\mathrm{chain}}$.
This specification preserves the non-tradability of crypto losses while allowing market participants to use
publicly available oracle updates when adjusting their financial hedges. Because the crypto loss factors are non-tradable, the market is \emph{incomplete}. This incompleteness is central to the valuation and
risk-minimizing hedging of crypto CAT bonds developed in subsequent sections.

\paragraph{Crypto Loss Dynamics and Oracle Observability}
\label{subsec:loss_oracle}
Let $\{X_{k,j}\}$ denote the incident-level loss observations reported to the smart contract by decentralized
oracles. For each monitoring period $k=1,\ldots,T$, let $N_k\in\mathbb{N}_0$ be the  number of loss-causing
incidents observed during period $t$, and let
\(
\{X_{k,j}\}_{j=1}^{N_k}
\)
be the corresponding incident losses posted on-chain (e.g., protocol exploit losses, bridge hack losses, or
smart-contract breach losses). We assume $X_{k,j}$ is $\mathcal F^{\mathrm{orc}}_k$-measurable for all $k$ and $j$, and
that the collection $\{X_{k,j}: 1\le k\le t,\ 1\le j\le N_k\}$ constitutes the sole loss input available to the
contract for trigger evaluation and settlement. We assume the oracle layer delivers these observed incident losses
using multiple oracle sources; accordingly, we treat $\{X_{k,j}\}$ as the primitive loss process governing the
contract, consistent with the information available to the on-chain settlement logic.
 
From the incident-level reports, we form the state variables that drive the trigger. Define the period-$k$ maximum
incident loss
\[
Y_k  = \max_{1\le j\le N_k} X_{k,j},
\qquad k=1,\ldots,T,
\]
with the convention $Y_k =0$ when $N_k=0$, and define the cumulative loss across periods by
\[
Z_k= \sum_{s=1}^{k}X_{s}, 
\qquad k=1,\ldots,T,
\]
where $X_s=\sum_{j=1}^{N_s} X_{s,j}$ is the cumulative loss during time $s$. Both $Y_k$ and $Z_k$ are $\mathcal F^{\mathrm{orc}}_k$-measurable and therefore observable by the smart contract at the
end of each monitoring period. Here $Y_k$ captures the largest single-incident shock within period $k$, while $Z_k$
records cumulative loss across periods.

\paragraph{Trigger Mechanism}
\label{subsec:trigger}

The trigger mechanism governs adjustments to coupon payments and potential
reductions of principal based on oracle-reported loss information. Trigger
evaluation is conducted at the end of each monitoring period and depends
exclusively on the observable loss statistics  $(Y_k, Z_k)$. Let $\{\delta^M_\xi\}_{\xi=0}^{m_1}$ and $\{\delta^C_\eta\}_{\eta=0}^{m_2}$ denote
predefined {\color{blue}increasing} threshold levels for the maximum-loss and cumulative-loss metrics,
respectively, with $\delta^M_0 = \delta^C_0 = 0$. For each monitoring period
$k$, the joint trigger state is defined as
\[
{\color{blue}A_{\xi_k}B_{\eta_k}}
=
\left\{
Y_k \in [\delta^M_{\xi-1}, \delta^M_\xi), \;
Z_k \in [\delta^C_{\eta-1}, \delta^C_\eta)
\right\},
\quad
\xi = 1,\ldots,m_1, \;
\eta = 1,\ldots,m_2.
\]
{\color{blue}Trigger activation occurs if either of the following two conditions is satisfied}:
(i) the maximum reported loss $Y_k$ exceeds its corresponding threshold level,
or
(ii) the maximum reported loss remains below its threshold while the cumulative
reported loss $Z_k$ exceeds the specified cumulative threshold.
These conditions allow the contract to respond both to abrupt, extreme loss
events and to sustained accumulation of moderate losses over time. By combining short-horizon severity through $Y_k$ with long-horizon
accumulation through $Z_k$, the joint trigger mitigates the limitations of
single-metric designs. In particular, it remains sensitive to sudden
catastrophic shocks such as large-scale protocol exploits, while also capturing
persistent deterioration that may not manifest through isolated extreme
events. This dual structure is well suited to the loss patterns observed in
crypto-native environments and provides a flexible basis for subsequent cash
flow and pricing specifications.

\paragraph{{\color{blue}Cash Flow Structure}} Let $F$ denote the face value of the bond. The cash flow received at monitoring
date $k$ is a deterministic function of the reference rate $R_k$ and the
oracle-reported trigger statistics $(Y_k,Z_k)$, and is given by
\begin{equation}
\label{eq:cashflow}
d(R_k,Y_k,Z_k)
=
\begin{cases}
F\, f_{\xi_k,\eta_k}(R_k), & k = 1,2,\ldots,T-1, \\[4pt]
F\, f_{\xi_k,\eta_k}(R_k) + F\, f_{\eta_k}(Z_k), & k = T,
\end{cases}
\end{equation}
{\color{blue}where \((\xi_k,\eta_k)\) denotes the realized trigger state at time \(k\), and $Y_k$ is associated with the the trigger $A_{\xi_k}$.} Coupon payments depend on both short-term and cumulative loss information,
while principal repayment is determined exclusively at maturity. The coupon adjustment function $f_{\xi,\eta}(\cdot)$ is specified as
\begin{equation}
\label{eq:coupon-with-spread}
f_{\xi_k,\eta_k}(R_k)
=
\Delta(A_{\xi_k}B_{\eta_k})\,(R_k + s),
\qquad
\xi = 1,\ldots,m_1,\;\eta = 1,\ldots,m_2,
\end{equation}
where $\Delta(A_{\xi_k}B_{\eta_k}) \in [0,\infty)$ denotes the coupon multiplier
associated with trigger state $A_{\xi_k}B_{\eta_k}$, and $s \in \mathbb{R}$ is a
contractually specified spread calibrated at issuance. This structure allows
coupon payments to scale with both prevailing market rates and the severity of
realized crypto losses.

Figure~\ref{fig:coupon} shows an illustrative example of the multiplier matrix with \( m_1 = 4 \) and \( m_2 = 3 \), represented as a heatmap.
 \begin{figure}[htb!]
    \centering
    \subfigure[Coupon]{
    \includegraphics[width=0.4\linewidth]{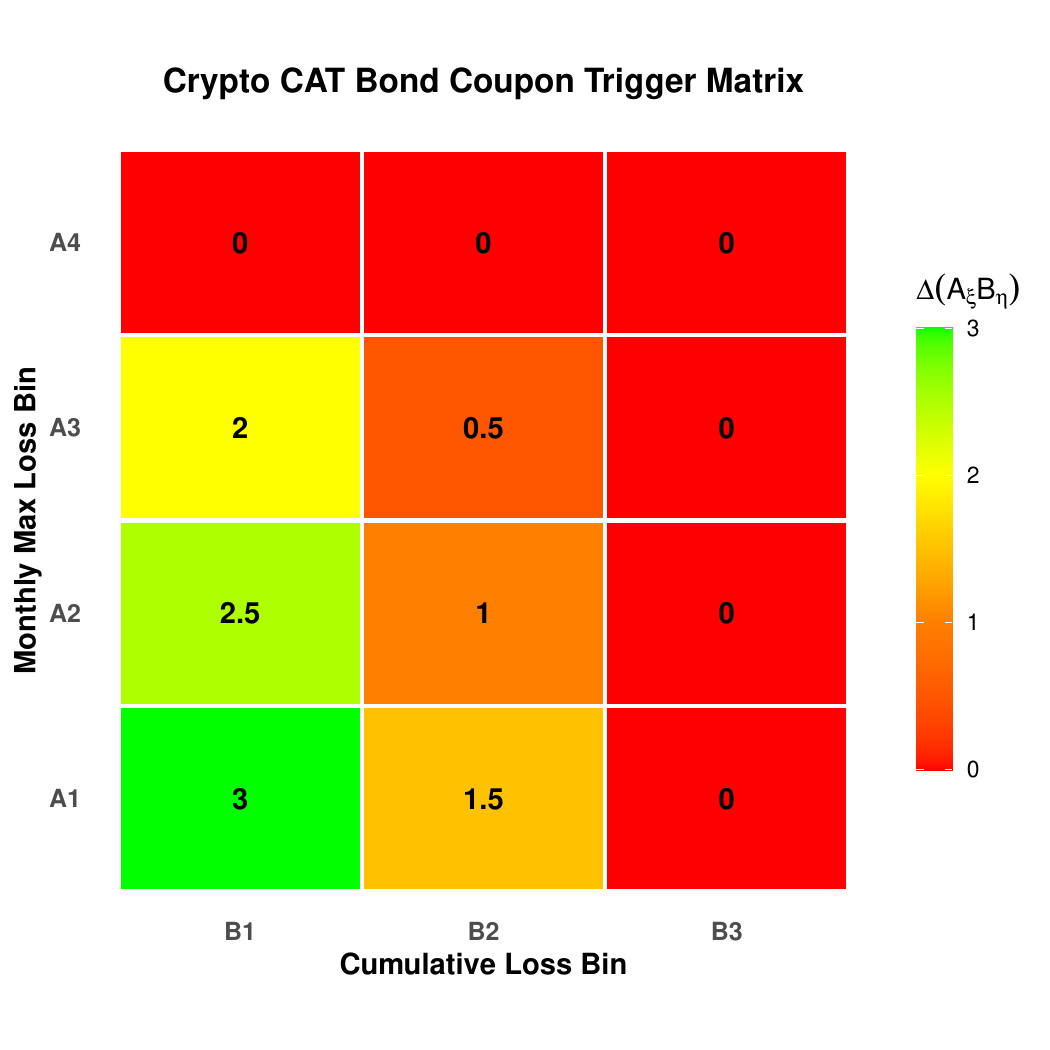} 
    \label{fig:coupon}}
    \subfigure[Principal]{
    \includegraphics[width=0.4\linewidth]{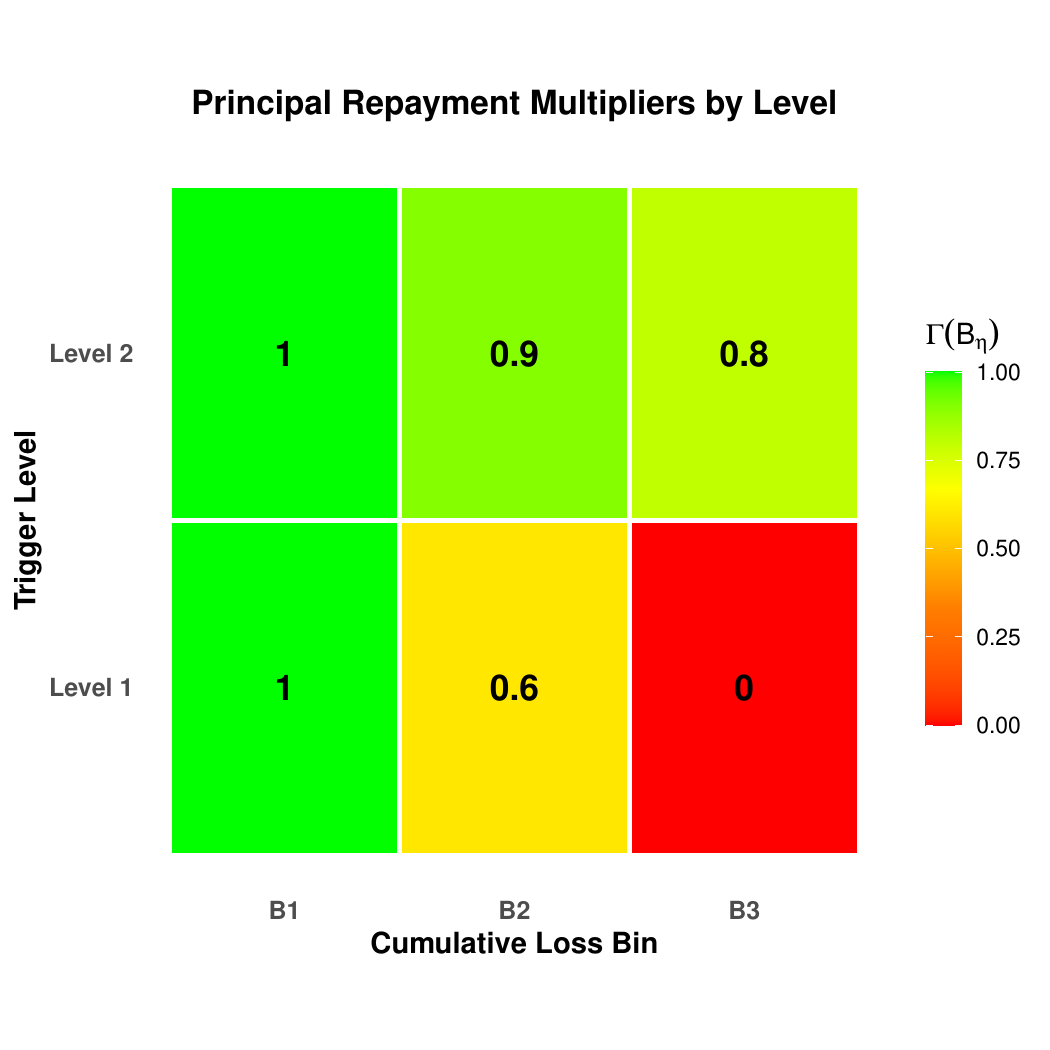}\label{fig:principle}}
    \caption{Crypto CAT bond coupon and  principal multipliers $\Delta(A_\xi B_\eta)$, and $\Gamma(B_\eta)$, respectively, $\xi=1,\ldots, 4$, and $\eta=1,2, 3$.}
 \end{figure}
The green-shaded cells on the bottom left of the heatmap correspond to favorable market conditions, where both cumulative losses and monthly maximum losses are low, yielding the highest coupon multipliers (e.g., \( \Delta = 3.0 \)). As one moves upward or to the right, reflecting increasingly severe losses in either dimension, the multiplier decreases progressively. The red-shaded region indicates the most adverse outcomes, where coupon payments are fully suspended. This severity-sensitive structure ensures that investor compensation scales with the underlying cryptosystem risk, aligning payouts with both short-term shock events and long-term systemic stress.

For the principal repayment function, it is determined solely by the cumulative loss. Specifically,
\begin{equation}\label{eq:principle}
    f_{\eta_T}(Z_T) = \Gamma(B_{\eta_T}),
\end{equation}
where  $\Gamma(B_{\eta_T})$ denotes the principal repayment multiplier associated with the cumulative loss bin \( B_{\eta_T} \), $\eta_T=1,2, 3$. To capture structural variation in bond design, we extend the principal repayment function to reflect multiple evaluation scenarios, each corresponding to a distinct trigger level. Figure~\ref{fig:principle} illustrates the repayment multipliers \( \Gamma(B_{\eta_T}) \) across cumulative loss bins \( B_1, B_2, B_3 \) for two such design levels. These levels do not represent tranches within a single bond, but rather alternative configurations that yield different bond valuations under stress. Level~1 reflects an aggressive structure: while it still provides full repayment in \(B_1\), repayment declines to 60\% in \(B_2\) and is entirely forfeited under severe cumulative losses in \(B_3\). By contrast, Level~2 is more conservative, with full repayment under low cumulative loss (\(B_1\)) and only a modest reduction ($90\%$) under moderate losses (\(B_2\)). This setup enables flexible scenario analysis and facilitates comparative pricing of alternative CAT bond designs tailored to different risk tolerances.

{
\paragraph{Arbitrage-Free Valuation and Pricing Measures}
\label{subsec:valuation}

We now formalize the valuation framework for crypto CAT bonds within the market
and information structure introduced above. 

In an incomplete market, the equivalent martingale measure is generally not
unique. To obtain a canonical pricing benchmark, we adopt the minimal
martingale measure (MMM) associated with the traded financial submarket
\citep{follmer1990hedging,mcneil2015quantitative}. Under the MMM, discounted
prices of tradable assets are martingales, while the distribution of
non-tradable crypto-loss and oracle risks remains unchanged (see Section \ref{sec:na_chain}). This choice yields
an arbitrage-free valuation framework that does not impose artificial risk
premia on sources of risk that cannot be hedged through financial trading. Let $\mathbb{Q}^{\mathrm{M}}$ denote the minimal martingale measure on
$(\Omega,\mathcal{F},\mathbb{F})$. The time-$0$ value of the crypto CAT bond is
given by the discounted expectation of its cash flows under
$\mathbb{Q}^{\mathrm{M}}$,
\[
V_0
=
\mathbb{E}^{\mathbb{Q}^{\mathrm{M}}}
\left[
\sum_{k=1}^{T}
\frac{d(R_k,Y_k,Z_k)}{D_k}
\right],
\]
provided the expectation is well defined, where $D_k$ denotes is the
money-market num\'eraire at time $k$ \citep{bjork2009arbitrage}. This valuation rule ensures consistency with
the absence of arbitrage in the traded financial submarket while fully
accounting for the impact of non-tradable crypto-loss dynamics and oracle-based
observability. The MMM also induces a natural decomposition of the bond
payoff into hedgeable and unhedgeable components. While the hedgeable component
can be partially offset through trading in financial assets, residual risk
arising from crypto losses must be borne by investors (see Section \ref{sec:na_chain} for details).

}

\paragraph{{\color{blue}On-Chain Settlement Architecture}} 
We next describe the on-chain settlement architecture through which trigger
evaluation and cash-flow execution are implemented, and which gives rise to
oracle-based observability in practice.

As illustrated in Figure~\ref{fig:cat_diagram}, the architecture is organized into four functional layers: \emph{Data and Pricing}, \emph{Trigger and Governance}, \emph{Capital Flow}, and \emph{Bond and Investor}. This framework serves as the execution engine for the pricing model derived in Section 2, ensuring that the cash flow function $d(R_k, Y_k, Z_k)$ is strictly enforced without requiring a central administrator. By embedding the calibrated payout maps $(\Delta,\Gamma)$ directly into the settlement logic, the system guarantees transparent, tamper-resistant execution, thereby reducing the agency costs and settlement latency inherent in off-chain issuances.

\begin{figure}[htb!]
    \centering
    \includegraphics[width=.8\textwidth]{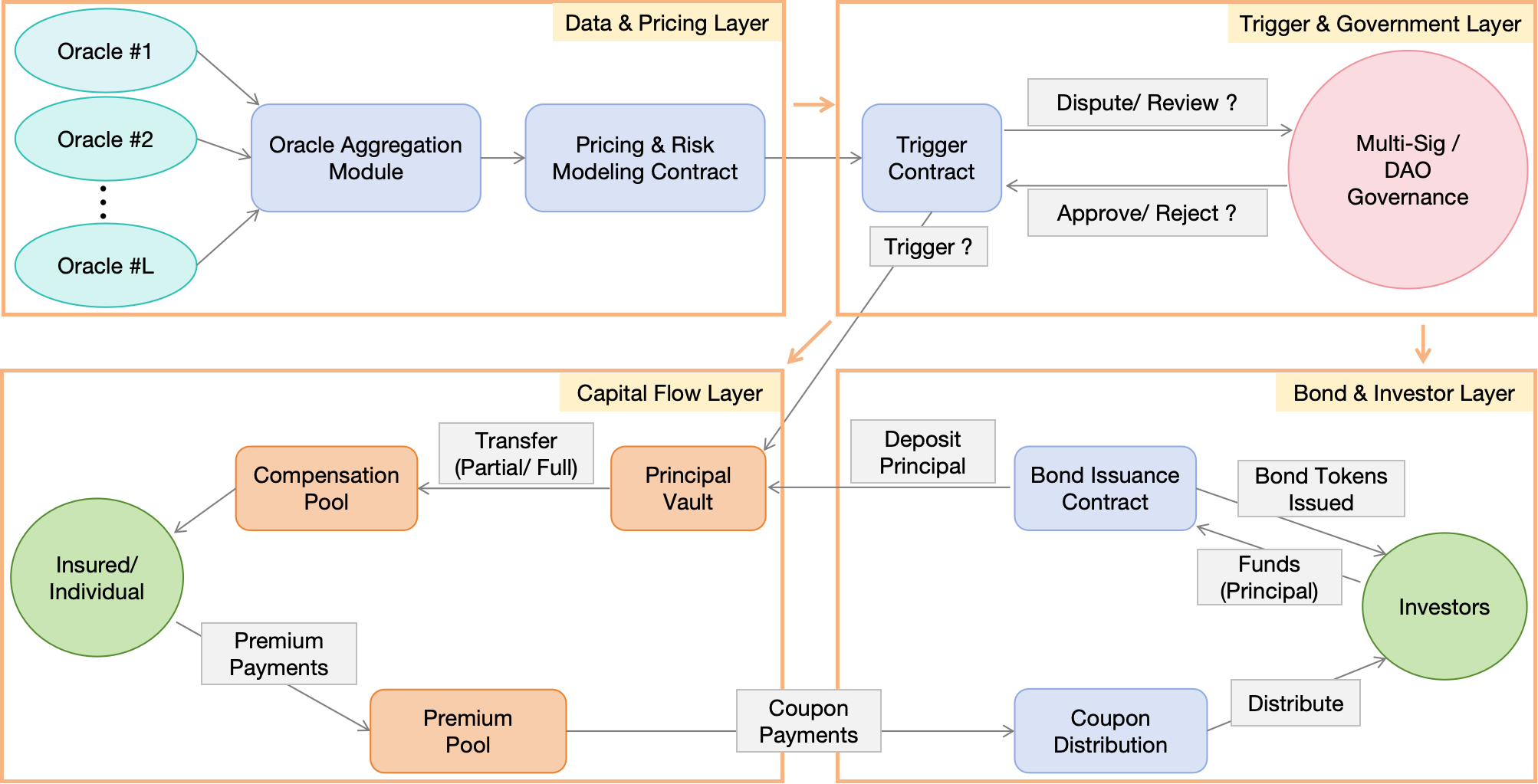}
    \caption{Operational architecture for trustless settlement. The automated flow of capital ensures full collateralization.}
    \label{fig:cat_diagram}
\end{figure}

We briefly describe the economic function of each layer.

\begin{itemize}
\item \textbf{Data and Pricing layer.}
This layer governs the acquisition and aggregation of loss-related information.
A decentralized oracle network aggregates event and loss data from multiple
independent sources, mitigating reliance on any single reporting entity. The
resulting oracle-reported loss metrics are passed to the pricing and risk
modeling contract, which encodes the pre-calibrated payout functions. At each
monitoring date, the realized trigger state $(Y_k,Z_k)$ and the corresponding
coupon multiplier are recorded on-chain, providing a publicly verifiable record
of the inputs determining bond cash flows.

\item \textbf{Trigger and Governance layer.}
The trigger contract maps the realized loss state $(Y_k,Z_k)$ to the
corresponding trigger cell $A_\xi B_\eta$, thereby determining coupon
adjustments and, at maturity, the applicable principal repayment multiplier
$\Gamma(B_{\eta_T})$. Under standard conditions, this mapping is deterministic and
automated. To accommodate rare boundary cases or data anomalies that may arise
in practice, an auxiliary governance mechanism may be invoked, allowing for
transparent resolution through predefined consensus procedures recorded
on-chain.

\item \textbf{Capital Flow layer.}
This layer manages the segregation and transfer of funds associated with the
bond. Capital raised at issuance is deposited into a principal vault that
functions as a collateral account backing investor claims, while coupon
payments are funded through a separate premium pool. Upon trigger activation,
the settlement logic reallocates a contractually specified portion of the
collateral to a compensation pool designated for sponsor recovery. At maturity,
the remaining balance in the principal vault is returned to investors, and any
accumulated compensation is released to the sponsor, enabling timely and
auditable settlement.

\item \textbf{Bond and Investor layer.}
The bond and investor layer governs the interaction between the contract and
market participants. Coupon payments are calculated based on the on-chain
reference rate and the applicable coupon multiplier, and are distributed
automatically according to contract rules. Tokenization of bond positions
facilitates ownership tracking and may support secondary market transfers,
subject to applicable regulatory and contractual constraints.
\end{itemize}

This architecture supports an end-to-end on-chain execution
process encompassing trigger evaluation, cash-flow determination, and
settlement, while remaining consistent with the information structure and
valuation framework developed earlier.

 \section{Theoretical Results}\label{sec:na_chain}
This section addresses RQ1--RQ3 by establishing an arbitrage-free pricing
framework under oracle-based observability, deriving a martingale-based risk
decomposition in the resulting incomplete market, and formulating sponsor-optimal
contract design under the dual measures $(\mathbb{Q},\mathbb{P})$.

\subsection{Arbitrage-free valuation under smart-contract constraints}

{
We work on a discrete-time probability space
$(\Omega,\mathcal{F},\mathbb{P})$ that supports three components of uncertainty:
(i) tradable financial market risk factors, (ii) oracle-reported crypto loss dynamics
relevant for contract settlement, and (iii) on-chain contract-state information.
Specifically, we consider the product space
\[
(\Omega,\mathcal{F},\mathbb{P})
=
(\Omega^{F}\times\Omega^{L}\times\Omega^{C},
\;\mathcal{F}^{F}\otimes\mathcal{F}^{L}\otimes\mathcal{F}^{C},
\;\mathbb{P}^{F}\otimes\mathbb{P}^{L}\otimes\mathbb{P}^{C}),
\]
{\color{blue}where the superscripts $F$, $L$, and $C$ correspond to financial market
risk factors, oracle-reported loss dynamics, and on-chain contract-state
information, respectively.} Let $\mathbb{F}^{\mathrm{fin}}=\{\mathcal{F}^F_k\}_{k=0}^T$ denote the filtration generated by the
tradable financial risk factors, and let $\mathbb{F}^{\mathrm{chain}}=\{\mathcal{F}^{\mathrm{chain}}_k\}_{k=0}^T$
denote the on-chain filtration defined by
\[
\mathcal{F}^{\mathrm{chain}}_k
=
\mathcal{F}^F_k \,\vee\, \mathcal{G}^O_k,
\qquad
\mathcal{G}^O_k
=
\sigma\big((Y_1,Z_1),\ldots,(Y_k,Z_k)\big)\,\vee\, \mathcal{F}^C_k,
\]
where $(Y_k,Z_k)$ are the oracle-reported loss statistics introduced in Section~2 and
$\mathcal{F}^C_k$ captures oracle commits and contract-state updates observable on-chain up to time $k$.
For notational convenience, write $\mathbb{F}^{\mathrm{orc}}=\{\mathcal{G}^O_k\}_{k=0}^T$, so that
$\mathbb{F}^{\mathrm{chain}}=\mathbb{F}^{\mathrm{fin}}\vee \mathbb{F}^{\mathrm{orc}}$.
 
We assume that the tradable financial risk factors are independent of the on-chain information relevant for
contract settlement. In our empirical implementation, the financial submarket is driven by macro factors
(Treasury rates, inflation, and the Secured Overnight Financing Rate (SOFR)), whereas crypto catastrophe
incidents (e.g., protocol exploits) are primarily driven by
endogenous blockchain mechanisms and vulnerabilities. Since these crypto incidents are contract-specific
and small relative to macro market depth, we adopt independence as a scale-separation approximation.
We acknowledge that independence is a convenient sufficient condition for theoretical tractability and may
not hold exactly during rare periods of joint stress; nevertheless, it provides a reasonable baseline in our
setting.

We assume the traded financial submarket is arbitrage free and admits a  MMM 
$\mathbb{Q}^{F}\sim \mathbb{P}^{F}$ such that all discounted traded asset prices are $\mathbb{Q}^{F}$-martingales
with respect to $\mathbb{F}^{\mathrm{fin}}$ \citep{jeanblanc2009mathematical}.
Trading is restricted to the traded financial assets. Accordingly, trading strategies are
assumed to be predictable with respect to the enlarged filtration $\mathbb{F}^{\mathrm{chain}}$.
Oracle publication and smart-contract state updates are treated as occurring at the monitoring dates, so that
strategies cannot condition on within-period intermediate execution states.  {\color{blue}Let $\{S_k^i\}_{k=0}^T$, $i=1,\ldots,d$, denote the price processes of $d$ traded financial assets, adapted to the
financial filtration $\mathbb F^{\mathrm{fin}}$.}    In the next theorem we show that the financial MMM extends naturally to the enlarged filtration and yields an
arbitrage-free valuation for $\mathbb{F}^{\mathrm{chain}}$-measurable cash flows.}
 
\begin{The}\label{thm:emm_chain_combined}
Define the product measure
\[
\mathbb{Q}=\mathbb{Q}^{F}\otimes \mathbb{P}^{L}\otimes \mathbb{P}^{C}.
\]
Then $\mathbb{Q}\sim\mathbb{P}$ on $(\Omega,\mathcal{F})$, and all discounted traded asset prices are
$\mathbb{Q}$-martingales with respect to the enlarged on-chain filtration
$\{\mathcal{F}^{\mathrm{chain}}_k\}_{k=0}^T$. Moreover, for any cash-flow stream $\{CF_k\}_{k=1}^T$ with $CF_k$ being
$\mathcal{F}^{\mathrm{chain}}_k$-measurable and satisfying
\[
\mathbb{E}^{\mathbb{Q}}
\!\left[
\sum_{k=1}^T \frac{|CF_k|}{D_k}
\right]
<\infty,
\]
the arbitrage-free ex-dividend value process is given by the risk-neutral pricing formula
\[
V_k
=
D_k\,
\mathbb{E}^{\mathbb{Q}}
\!\left[
\sum_{t=k+1}^T \frac{CF_t}{D_t}\,\middle|\,
\mathcal{F}^{\mathrm{chain}}_k
\right],
\qquad k=0,\ldots,T.
\]
\end{The}

\begin{proof}
Since $\mathbb{Q}^{F}\sim\mathbb{P}^{F}$ on $(\Omega^{F},\mathcal{F}^{F})$ and
$\mathbb{P}^{L},\mathbb{P}^{C}$ are unchanged, the product measure
$\mathbb{Q}=\mathbb{Q}^{F}\otimes\mathbb{P}^{L}\otimes\mathbb{P}^{C}$ satisfies
$\mathbb{Q}\sim\mathbb{P}$ on $(\Omega,\mathcal{F})$. Let $\tilde S^i_k=S^i_k/D_k$
be a discounted traded asset price, which is $\mathcal{F}^F_k$-adapted. Then, for
$k=0,\ldots,T-1$,
\[
\mathbb{E}^{\mathbb{Q}}\!\left[\tilde S^i_{k+1}\mid \mathcal{F}^{\mathrm{chain}}_k\right]
=
\mathbb{E}^{\mathbb{Q}}\!\left[\tilde S^i_{k+1}\mid \mathcal{F}^{F}_k\right]
=
\mathbb{E}^{\mathbb{Q}^{F}}\!\left[\tilde S^i_{k+1}\mid \mathcal{F}^{F}_k\right]
=
\tilde S^i_k,
\]
where the first equality uses that $\tilde S^i_{k+1}$ depends only on the
$\Omega^F$ coordinate (hence the additional on-chain information is irrelevant),
and the second uses that $\mathbb{Q}$ coincides with $\mathbb{Q}^F$ on the
financial component. Thus discounted traded prices are
$(\mathbb{Q},\mathcal{F}^{\mathrm{chain}}_k)$-martingales.

Next, let $\{CF_k\}_{k=1}^T$ be an $\mathbb{F}^{\mathrm{chain}}$-adapted cash-flow stream such that
$\mathbb{E}^{\mathbb{Q}}[\sum_{k=1}^T |CF_k|/D_k]<\infty$, and define the ex-dividend discounted value process by
\[
\tilde V_k
=
\mathbb{E}^{\mathbb{Q}}
\!\left[
\sum_{t=k+1}^T \frac{CF_t}{D_t}\,\middle|\,
\mathcal{F}^{\mathrm{chain}}_k
\right],
\qquad k=0,\ldots,T.
\]
By integrability, $\tilde V_k$ is finite a.s., and the time-$k$ ex-dividend value is then
\(
V_k =D_k\,\tilde V_k.
\)
This completes the proof.
  \qed
\end{proof}

{\bf Remark:} Theorem~\ref{thm:emm_chain_combined} answers RQ1 by establishing an arbitrage-free valuation framework for
on-chain-settled crypto CAT bonds whose cash flows are measurable with respect to the enlarged on-chain
filtration. In particular, the product-form pricing measure $\mathbb Q$ extends the financial submarket's martingale
measure to $\mathbb F^{\mathrm{chain}}$, and the bond price is obtained by the standard risk-neutral discounted
conditional expectation under $\mathbb Q$.

\medskip
While Theorem~\ref{thm:emm_chain_combined} establishes arbitrage-free valuation,
it does not by itself characterize which sources of risk can be hedged through
financial trading. As mentioned before, since the market is incomplete, the EMM is not unique. In what follows, we adopt the MMM $\mathbb{Q}^M$ in the sense of \citet{follmer1990hedging, follmer2010minimal}:
among all pricing measures, we select the one that changes measure only on the
tradable financial coordinate and leaves the law of the non-traded sources of
risk unchanged. Equivalently, we work with the extension of the financial MMM to the full product space, i.e.,
\(
\mathbb{Q}^{M}
=
\mathbb{Q}^{F}\otimes \mathbb{P}^{L}\otimes \mathbb{P}^{C}.
\) For brevity, we write $\mathbb{Q}=\mathbb{Q}^{M}$ throughout the remainder of this section.

 \paragraph{Discounted payoff and hedgeable projection.}
Let $\mathbb Q=\mathbb Q^{M}$ denote the extended minimal martingale measure on the full product space, so that
the discounted traded asset prices $\tilde{\mathbf S}=(\tilde S^1,\ldots,\tilde S^d)$ form a square-integrable
$\mathbb Q$-martingale with respect to the enlarged filtration
$\mathbb F^{\mathrm{chain}}=\{\mathcal F^{\mathrm{chain}}_k\}_{k=0}^T$.
Let the  discounted  total payoff be
\(
H = \sum_{t=1}^{T} \frac{CF_t}{D_t} \in L^2(\mathbb Q).
\)
Define the discounted terminal gain space generated by self-financing trading in traded assets as
\[
\mathcal K
 =
\left\{
\sum_{k=0}^{T-1}\boldsymbol\vartheta_k^\top(\tilde{\mathbf S}_{k+1}-\tilde{\mathbf S}_k)
:\ 
\boldsymbol\vartheta_k\in L^0(\mathcal F^{\mathrm{chain}}_k;\mathbb R^d),\
\sum_{k=0}^{T-1}\boldsymbol\vartheta_k^\top(\tilde{\mathbf S}_{k+1}-\tilde{\mathbf S}_k)\in L^2(\mathbb Q)
\right\},
\]
and let $\overline{\mathcal K}$ be its $L^2(\mathbb Q)$-closure. Set
\(
\mathcal H  = \mathbb R \oplus \overline{\mathcal K},
\)
which is a closed linear subspace of $L^2(\mathbb Q)$.
Define the mean-square hedgeable projection and residual by
\[
H^{F} =\operatorname{Proj}_{\mathcal H}(H),\qquad H^\perp =H-H^{F}.
\]

\begin{The}\label{thm:decomp_contract_design}
Let $\tilde V_k$ be the  {ex-dividend} discounted value process  as in Theorem~\ref{thm:emm_chain_combined}.
Then
\[
\tilde V_k = \tilde V_k^{F} + M_k,
\qquad
\tilde V_k^{F} =\mathbb E^{\mathbb Q}\!\left[H^{F}\,\middle|\,\mathcal F^{\mathrm{chain}}_k\right]
-\sum_{t=1}^k \frac{CF_t}{D_t},
\qquad
M_k =\mathbb E^{\mathbb Q}\!\left[H^{\perp}\,\middle|\,\mathcal F^{\mathrm{chain}}_k\right].
\]
Moreover, for each traded asset $i=1,\ldots,d$ and all $k=0,\ldots,T-1$,
\[
\mathbb E^{\mathbb Q}\!\left[(M_{k+1}-M_k)\,(\tilde S^i_{k+1}-\tilde S^i_k)\,\middle|\,\mathcal F^{\mathrm{chain}}_k\right]=0.
\]
Consequently, $M$ represents the component of contract risk that cannot be eliminated in mean-square sense by any
self-financing trading strategy in the traded financial submarket.
\end{The}

\begin{proof}
Recall $H=\sum_{t=1}^T CF_t/D_t \in L^2(\mathbb Q)$ and the decomposition $H=H^F+H^\perp$.
Since $\sum_{t=1}^k CF_t/D_t$ is $\mathcal F^{\mathrm{chain}}_k$-measurable, we have
\begin{align*}
\tilde V_k
&=
\mathbb{E}^{\mathbb{Q}}
\!\left[
\sum_{t=k+1}^T \frac{CF_t}{D_t}\,\middle|\,
\mathcal{F}^{\mathrm{chain}}_k
\right]
=
\mathbb E^{\mathbb Q}\!\left[H-\sum_{t=1}^k \frac{CF_t}{D_t}\,\middle|\,\mathcal F^{\mathrm{chain}}_k\right] \\
&=
\mathbb E^{\mathbb Q}\!\left[H\,\middle|\,\mathcal F^{\mathrm{chain}}_k\right]
-\sum_{t=1}^k \frac{CF_t}{D_t}
=
\mathbb E^{\mathbb Q}\!\left[H^{F}\,\middle|\,\mathcal F^{\mathrm{chain}}_k\right]
+\mathbb E^{\mathbb Q}\!\left[H^{\perp}\,\middle|\,\mathcal F^{\mathrm{chain}}_k\right]
-\sum_{t=1}^k \frac{CF_t}{D_t} \\
&=
\left(\mathbb E^{\mathbb Q}\!\left[H^{F}\,\middle|\,\mathcal F^{\mathrm{chain}}_k\right]
-\sum_{t=1}^k \frac{CF_t}{D_t}\right)
+
\mathbb E^{\mathbb Q}\!\left[H^{\perp}\,\middle|\,\mathcal F^{\mathrm{chain}}_k\right]
=
\tilde V_k^{F}+M_k.
\end{align*}

We next prove orthogonality. Let $\Delta \tilde S^i_{k+1}=\tilde S^i_{k+1}-\tilde S^i_k$.
Using
\[
M_{k+1}-M_k
=
\mathbb{E}^{\mathbb{Q}}\!\left[H^\perp\mid\mathcal{F}^{\mathrm{chain}}_{k+1}\right]
-
\mathbb{E}^{\mathbb{Q}}\!\left[H^\perp\mid\mathcal{F}^{\mathrm{chain}}_{k}\right],
\]
we compute
\begin{align*}
\mathbb{E}^{\mathbb{Q}}\!\left[(M_{k+1}-M_k)\Delta \tilde S^i_{k+1}\mid\mathcal{F}^{\mathrm{chain}}_k\right]
&=
\mathbb{E}^{\mathbb{Q}}\!\left[
\mathbb{E}^{\mathbb{Q}}\!\left[H^\perp\mid\mathcal{F}^{\mathrm{chain}}_{k+1}\right]\Delta \tilde S^i_{k+1}
\mid\mathcal{F}^{\mathrm{chain}}_k\right]\\
&\quad-
\mathbb{E}^{\mathbb{Q}}\!\left[
\mathbb{E}^{\mathbb{Q}}\!\left[H^\perp\mid\mathcal{F}^{\mathrm{chain}}_{k}\right]\Delta \tilde S^i_{k+1}
\mid\mathcal{F}^{\mathrm{chain}}_k\right].
\end{align*}
For the first term, since $\Delta \tilde S^i_{k+1}$ is $\mathcal{F}^{\mathrm{chain}}_{k+1}$-measurable,
\[
\mathbb{E}^{\mathbb{Q}}\!\left[
\mathbb{E}^{\mathbb{Q}}\!\left[H^\perp\mid\mathcal{F}^{\mathrm{chain}}_{k+1}\right]\Delta \tilde S^i_{k+1}
\mid\mathcal{F}^{\mathrm{chain}}_k\right]
=
\mathbb{E}^{\mathbb{Q}}\!\left[
H^\perp\,\Delta \tilde S^i_{k+1}\mid\mathcal{F}^{\mathrm{chain}}_k\right].
\]
For the second term, $\mathbb{E}^{\mathbb{Q}}[H^\perp\mid\mathcal{F}^{\mathrm{chain}}_{k}]$ is
$\mathcal{F}^{\mathrm{chain}}_k$-measurable, hence it can be pulled out:
\[
\mathbb{E}^{\mathbb{Q}}\!\left[
\mathbb{E}^{\mathbb{Q}}\!\left[H^\perp\mid\mathcal{F}^{\mathrm{chain}}_{k}\right]\Delta \tilde S^i_{k+1}
\mid\mathcal{F}^{\mathrm{chain}}_k\right]
=
\mathbb{E}^{\mathbb{Q}}\!\left[H^\perp\mid\mathcal{F}^{\mathrm{chain}}_{k}\right]
\mathbb{E}^{\mathbb{Q}}\!\left[\Delta \tilde S^i_{k+1}\mid\mathcal{F}^{\mathrm{chain}}_k\right].
\]
Since $\tilde S^i$ is a $\mathbb{Q}$-martingale with respect to the enlarged filtration, we have
$\mathbb{E}^{\mathbb{Q}}[\Delta \tilde S^i_{k+1}\mid\mathcal{F}^{\mathrm{chain}}_k]=0$,
and thus the second term vanishes. Therefore,
\[
\mathbb{E}^{\mathbb{Q}}\!\left[(M_{k+1}-M_k)\Delta \tilde S^i_{k+1}\mid\mathcal{F}^{\mathrm{chain}}_k\right]
=
\mathbb{E}^{\mathbb{Q}}\!\left[
H^\perp\,\Delta \tilde S^i_{k+1}\mid\mathcal{F}^{\mathrm{chain}}_k\right].
\]
By construction,
$H^\perp\perp\mathcal H$ and $\mathcal H$ contains the one-step gains
$1_A\,\Delta\tilde S^i_{k+1}$ for every $A\in\mathcal F^{\mathrm{chain}}_k$.
Hence
\[
\mathbb E^{\mathbb Q}\!\left[H^\perp\,1_A\,\Delta \tilde S^i_{k+1}\right]=0,
\qquad \forall\,A\in\mathcal F^{\mathrm{chain}}_k.
\]
This is equivalent to
\[
\mathbb E^{\mathbb Q}\!\left[H^\perp\,\Delta \tilde S^i_{k+1}\,\middle|\,\mathcal F^{\mathrm{chain}}_k\right]=0,
\]
which completes the proof. \qed
\end{proof}

{\bf Remark:} The decomposition result answers RQ2 by showing that crypto CAT bond values
naturally split into a hedgeable component driven by tradable financial risk
factors and a residual component driven by non-tradable crypto-loss and oracle
execution risk. Under the MMM, only the hedgeable
component is affected by financial trading, while the residual component
remains orthogonal to all admissible trading strategies. This separation
provides a canonical risk decomposition and clarifies which sources of risk can
be mitigated through market instruments and which must be borne by investors.

Theorem~\ref{thm:decomp_contract_design} implies that contract parameters
$(\Delta,\Gamma)$ influence not only the price level via $\mathbb{E}^{\mathbb{Q}}[H]$,
but also the non-hedgeable  risk component $H^\perp$ that investors
must bear. This motivates incorporating risk-sensitive constraints (e.g.,  TVaR
or tail penalties) when optimizing the payout maps under sponsor objectives.

\subsection{Sponsor-optimal contract design}
\label{sec:rq3_tail}
This section formulates sponsor-optimal bond
design with particular emphasis on tail risk. Throughout, we treat the issuance size as exogenous. Specifically, the sponsor
issues a fixed number of notes $N\in\mathbb{N}$ with per-note face value $F>0$,
so that the total collateralized principal is $F^{\mathrm{tot}} = N F$.
Accordingly, the design problem studied here concerns the {shape} of the
on-chain payout schedule (i.e., $(\Delta,\Gamma)$), while the total risk
transfer limit $F^{\mathrm{tot}}$ is taken as given. As in Section~\ref{sec:framework}, a crypto CAT bond design is fully determined
by the on-chain payout maps
\[
\Delta=\{\Delta(A_\xi B_\eta)\}_{\xi=1,\ldots,m_1;\,\eta=1,\ldots,m_2},
\qquad
\Gamma=\{\Gamma(B_\eta)\}_{\eta=1,\ldots,m_2},
\]
which specifies coupon and principal multipliers on the joint trigger grid.  

\paragraph{Sponsor loss and tail objective under $\mathbb{P}$.}
Let $L=Z_T$ denote the sponsor’s aggregate economic loss over the bond horizon,
modeled under the physical measure $\mathbb{P}$. The crypto CAT bond provides
protection through the on-chain settlement logic described in
Section~\ref{sec:framework}: a portion of the collateralized principal is
released from the principal vault to the sponsor’s compensation pool when the
terminal trigger state is adverse. Under the maturity principal rule Eq. \eqref{eq:principle},
the realized principal repayment multiplier is $\Gamma(B_{\eta_T})$, where
$B_{\eta_T}$ is the cumulative-loss bin determined by $Z_T$. Accordingly, we
define the sponsor’s realized protection as
\begin{equation}\label{eq:protection_def}
P_N(\Delta,\Gamma)
=
F^{\mathrm{tot}}\Bigl(1-\Gamma(B_{\eta_T})\Bigr)
=
NF\Bigl(1-\Gamma(B_{\eta_T})\Bigr).
\end{equation}
We then define the sponsor’s net shortfall as
\[
X_N(\Delta,\Gamma)=\bigl(L-P_N(\Delta,\Gamma)\bigr)_{+}.
\]
To focus on extreme outcomes, we adopt a tail-risk objective and measure
sponsor exposure by the TVaR at confidence level $\alpha\in(0,1)$:
\[
J_\alpha(\Delta,\Gamma)
=
\mathrm{TVaR}^{\mathbb{P}}_{\alpha}\!\left(X_N(\Delta,\Gamma)\right).
\]
This choice is natural in catastrophe-risk applications and directly targets
the sponsor’s residual exposure in severe loss scenarios. For a given design $(\Delta,\Gamma)\in\mathcal{D}$ and coupon spread $s$, the
time-$0$ \emph{per-note} market value of the bond under the pricing measure
$\mathbb{Q}$ is
\[
V_0(\Delta,\Gamma;s)
=
\mathbb{E}^{\mathbb{Q}}
\!\left[
\sum_{k=1}^{T}\frac{CF_k(\Delta,\Gamma;s)}{D_k}
\right],
\]
where $CF_k(\Delta,\Gamma;s)$ follows the cash-flow rule in Eqs.
\eqref{eq:cashflow}--\eqref{eq:principle} and depends on $s$ through the coupon
rate $(R_k+s)$. To ensure investor marketability, we impose a par-issuance
condition and calibrate the spread via
\begin{equation}\label{eq:par_tail}
V_0(\Delta,\Gamma;s)=F,
\end{equation}
where $F$ denotes the per-note face value.  

The sponsor chooses an admissible design and corresponding calibrated spread to
minimize tail risk under $\mathbb{P}$ while satisfying market valuation under
$\mathbb{Q}$:
\begin{equation}\label{eq:tail_design_problem}
\min_{(\Delta,\Gamma)\in\mathcal{D},\, s\in\mathbb{R}}
\quad
\mathrm{TVaR}^{\mathbb{P}}_{\alpha}\!\left(X_N(\Delta,\Gamma)\right)
\quad
\text{s.t.}\quad
V_0(\Delta,\Gamma;s)=F.
\end{equation}
This formulation makes explicit the dual-measure nature of crypto CAT bond
design: the pricing measure $\mathbb{Q}$ governs investor participation through
the par constraint, while the physical measure $\mathbb{P}$ governs the
sponsor’s exposure to extreme crypto losses.

 In the following, we prove the existence of sponsor-optimal tail design.

\begin{The}\label{thm:tail_exist}
Fix $N\in\mathbb{N}$ and $F>0$, and let $P_N(\Delta,\Gamma)$ be defined by
Eq. \eqref{eq:protection_def}. Assume $L\in L^1(\mathbb{P})$ and that $\mathcal D$ is compact and satisfies
\(
\inf_{(\Delta,\Gamma)\in\mathcal D} \mathbb E^{\mathbb Q}\!\big[\sum_{k=1}^T  \Delta(A_{\xi_k}B_{\eta_k}) /{D_k}\big]\;\ge\;\underline V \;>\;0.
\)
Then the tail-risk design problem
Eq. \eqref{eq:tail_design_problem} admits at least one optimal solution
$(\Delta^\star,\Gamma^\star,s^\star)$.
\end{The}

\begin{proof}
Fix $N\in\mathbb{N}$ and $F>0$. Since $0\le P_N(\Delta,\Gamma)\le NF$, we have
\[
0\le X_N(\Delta,\Gamma)\le L+NF,
\qquad (\Delta,\Gamma)\in\mathcal D,
\]
and thus $X_N(\Delta,\Gamma)\in L^1(\mathbb P)$ for all $(\Delta,\Gamma)\in\mathcal D$ because $L\in L^1(\mathbb P)$.
Consequently, $J_\alpha(\Delta,\Gamma)=\mathrm{TVaR}^{\mathbb P}_\alpha(X_N(\Delta,\Gamma))$ is finite on $\mathcal D$.

We next show that $(\Delta,\Gamma)\mapsto J_\alpha(\Delta,\Gamma)$ is continuous on $\mathcal D$.
Let $(\Delta_n,\Gamma_n)\to(\Delta,\Gamma)$ in $\mathcal D$. By the definition of $P_N(\Delta,\Gamma)$
 and the boundedness of multipliers on $\mathcal D$, we have
$P_N(\Delta_n,\Gamma_n)\to P_N(\Delta,\Gamma)$ almost surely, hence
$X_N(\Delta_n,\Gamma_n)\to X_N(\Delta,\Gamma)$ almost surely. Moreover,
\[
|X_N(\Delta_n,\Gamma_n)-X_N(\Delta,\Gamma)| \le X_N(\Delta_n,\Gamma_n)+X_N(\Delta,\Gamma)\le 2(L+NF)\in L^1(\mathbb P),
\]
so dominated convergence yields $X_N(\Delta_n,\Gamma_n)\to X_N(\Delta,\Gamma)$ in $L^1(\mathbb P)$ \citep{billingsley2013convergence}.

Using the dual representation of $\mathrm{TVaR}$, for any $X\in L^1(\mathbb P)$,
\[
\mathrm{TVaR}^{\mathbb P}_\alpha(X)
=
\sup_{Z\in\mathcal Z_\alpha}\mathbb E^{\mathbb P}[Z\,X],
\qquad
\mathcal Z_\alpha=\Big\{Z\in L^\infty:\ 0\le Z\le \tfrac{1}{1-\alpha}\ \text{a.s.},\ \mathbb E^{\mathbb P}[Z]=1\Big\}.
\]
Therefore, for any $Z\in\mathcal Z_\alpha$,
\[
\Big|\mathbb E^{\mathbb P}[Z\,X_N(\Delta_n,\Gamma_n)]-\mathbb E^{\mathbb P}[Z\,X_N(\Delta,\Gamma)]\Big|
\le \frac{1}{1-\alpha}\,\mathbb E^{\mathbb P}\!\big|X_N(\Delta_n,\Gamma_n)-X_N(\Delta,\Gamma)\big|.
\]
Taking the supremum over $Z\in\mathcal Z_\alpha$ yields
\[
|J_\alpha(\Delta_n,\Gamma_n)-J_\alpha(\Delta,\Gamma)|
\le \frac{1}{1-\alpha}\,\mathbb E^{\mathbb P}\!\big|X_N(\Delta_n,\Gamma_n)-X_N(\Delta,\Gamma)\big|
\to 0,
\]
so $J_\alpha$ is continuous on $\mathcal D$.

Under the cashflow rule Eqs. \eqref{eq:cashflow}--\eqref{eq:principle}, the spread $s$ enters linearly through the coupon
rate $(R_k+s)$, hence for each fixed $(\Delta,\Gamma)\in\mathcal D$ the time-$0$ value is affine in $s$:
\[
V_0(\Delta,\Gamma;s)=V_0(\Delta,\Gamma;0)+s\,V(\Delta),
\]
where
\[
V(\Delta)=
\mathbb{E}^{\mathbb{Q}}
\!\left[
\sum_{k=1}^{T}\frac{F\,\Delta(A_{\xi_k}B_{\eta_k})}{D_k}
\right].
\]
Under the integrability condition $\mathbb{E}^{\mathbb{Q}}[\sum_{k=1}^{T}(R_k+1)/D_k]<\infty$ and boundedness of the
payout multipliers on $\mathcal D$, dominated convergence implies that $V_0(\Delta,\Gamma;0)$ and $V(\Delta)$ are finite
and continuous in $(\Delta,\Gamma)$. Therefore, for each $(\Delta,\Gamma)\in\mathcal D$ the par condition $V_0(\Delta,\Gamma;s)=F$ admits a
unique solution
\begin{equation}\label{eq:spead}
s(\Delta,\Gamma)=\frac{F-V_0(\Delta,\Gamma;0)}{V(\Delta)}.
\end{equation}
By continuity of $V_0(\Delta,\Gamma;0)$ and $V(\Delta)$ and the uniform lower bound on $V(\Delta)$, the mapping
$(\Delta,\Gamma)\mapsto s(\Delta,\Gamma)$ is continuous. Since $\mathcal D$ is compact, $s(\Delta,\Gamma)$ is bounded on
$\mathcal D$, and thus the feasible set
\[
\mathcal{F}
=
\{(\Delta,\Gamma,s)\in\mathcal{D}\times\mathbb{R}:\ V_0(\Delta,\Gamma;s)=F\}
\]
is compact  \citep{rudin1987real}. Define $\widetilde{J}_\alpha(\Delta,\Gamma,s)=J_\alpha(\Delta,\Gamma)$ on $\mathcal F$. Since $\widetilde{J}_\alpha$ is
continuous  on the compact set $\mathcal F$, it attains its minimum. Hence there exists
$(\Delta^\star,\Gamma^\star,s^\star)\in\mathcal F$ solving Eq. \eqref{eq:tail_design_problem}. \qed
\end{proof}

{\bf Remark.} The existence result formalizes sponsor-side design under dual
measures: within the class of smart-contract-feasible grid-based designs
$(\Delta,\Gamma)\in\mathcal{D}$ and fixed issuance size $N$, the sponsor
minimizes tail shortfall under the physical measure $\mathbb{P}$, while the
market-clearing spread is determined endogenously by the par constraint under
the pricing measure $\mathbb{Q}$. In particular, the result ensures that an
optimal payout schedule exists even though valuation and risk objectives are
evaluated under distinct measures.

In the empirical study, we restrict attention to a finite family of
smart-contract-feasible designs
\[
\mathcal{D}_N=\{(\Delta^{(j)},\Gamma^{(j)}): j=1,\ldots,N\}\subset\mathcal{D},
\]
constructed by varying trigger thresholds and payout maps. For each candidate
design $(\Delta,\Gamma)\in\mathcal{D}_N$, the coupon spread $s(\Delta,\Gamma)$ is
calibrated by the par condition
\(
V_0(\Delta,\Gamma;s)=F,
\)
yielding a market-clearing contract under the pricing measure $\mathbb{Q}$.
}
\section{Statistical Calibration of Loss and Discounting Dynamics} \label{sec:statistics1}
While Sections~\ref{sec:framework}--\ref{sec:na_chain} develop the pricing and
design theory for crypto CAT bonds, the practical performance of any proposed
contract depends on how accurately the underlying crypto-loss and financial
risk drivers are characterized. This section aims to answer RQ4 by introducing a data-driven statistical model that captures
the salient empirical properties of both components: heavy-tailed
crypto losses, and time-varying market conditions. We begin with exploratory evidence
to motivate the modeling choices, then estimate the marginal dynamics and
dependence structure used to simulate joint paths and compute (i) market values
under $\mathbb{Q}$ and (ii) sponsor tail exposure under $\mathbb{P}$ in the
subsequent section.

\subsection{Exploratory Data Analysis}\label{sec:eda}
\subsubsection{Crypto Risks}
 One of the most comprehensive sources of DeFi incident data is the REKT database, which documents security events with details such as the occurrence date, affected blockchain, vulnerability type, and loss amount. The dataset analyzed in this study covers the period from January 2020 to December 2023, providing a representative view of the evolving threat landscape in the DeFi ecosystem. Data from 2024 is reserved for assessment and replay analysis.  The data is categorized into three groups based on the underlying blockchain: ETH,  BSC, and Other.  This categorization serves two purposes. First, it reflects the dominant ecosystems within the DeFi space:  ETH  and  BSC  have consistently hosted the largest and most active DeFi platforms, while the Other category aggregates incidents from smaller or emerging blockchains. Second, grouping incidents by blockchain type allows us to capture different risk profiles, smart contract practices, and ecosystem maturity levels, critical factors for designing catastrophe bond structures tailored to specific risk pools.  The ETH series exhibits gaps between January and July 2020; we therefore restrict it to August 2020–December 2023, yielding 41 months. The BSC series is continuous from October 2020 to December 2023, yielding 39 months of data. The Other series runs from June 2021 to December 2023, comprising 31 months. 

\begin{table}[htb!]
    \centering
    \caption{Summary statistics of monthly maximum losses and {monthly aggregate loss} of crypto incidents from the REKT database, where `SD' denotes standard deviation, and `Log' represents log scaled loss.}
    \label{tab:monthly_maximum_loss_long}
    \resizebox{\textwidth}{!}{
    \begin{tabular}{llccccccc}
        \toprule
        Category & Scale & Min  & Median & Mean & Max  & SD & Count & Period \\
        \midrule
        \multicolumn{9}{c}{{\color{blue}Monthly Maximum Loss}}\\
        \midrule
        \multirow{2}{*}{ETH}
        & Original & $4.94\times10^5$ & $2.74\times10^7$ & $7.21\times10^7$ & $6.02\times10^8$ & $1.08\times10^8$ 
        & \multirow{2}{*}{41} & \multirow{2}{*}{8/2020--12/2023} \\
        & Log      & 13.11 & 17.13 & 16.84 & 20.22 & 1.94
        & & \\
        \midrule
        \multirow{2}{*}{BSC}
        & Original & $1.11\times10^5$ & $3.68\times10^6$ & $6.50\times10^7$ & $1.30\times10^9$ & $2.24\times10^8$ 
        & \multirow{2}{*}{39} & \multirow{2}{*}{10/2020--12/2023} \\
        & Log       & 11.62 & 15.12 & 15.62 & 20.99 & 2.16 
        & & \\
        \midrule
        \multirow{2}{*}{Other}
        & Original & $2.64\times10^5$ & $3.60\times10^7$ & $4.10\times10^8$ & $3.60\times10^9$ & $8.45\times10^8$ 
        & \multirow{2}{*}{31} & \multirow{2}{*}{6/2021--12/2023} \\
        & Log     & 12.48 & 17.40 & 17.59 & 22.00 & 2.50 
        & & \\
           \midrule
        \multicolumn{9}{c}{{Monthly Aggregate Loss}}\\
        \midrule
         \multirow{2}{*}{ETH}
        & Original & $1.12\times10^6$ & $5.19\times10^7$ & $1.05\times10^8$ & $6.96\times10^8$ & $1.44\times10^8$ 
        & \multirow{2}{*}{41} & \multirow{2}{*}{8/2020--12/2023} \\
        & Log       & 13.93 & 17.76 & 17.43 & 20.36 & 1.71 
        & & \\
        \midrule
        \multirow{2}{*}{BSC}
        & Original & $1.11\times10^5$ & $1.01\times10^7$ & $7.59\times10^7$ & $1.40\times10^9$ & $2.39\times10^8$ 
        & \multirow{2}{*}{39} & \multirow{2}{*}{10/2020--12/2023} \\
        & Log       & 11.62 & 16.12 & 16.23 & 21.06 & 1.99 
        & & \\
        \midrule
        \multirow{2}{*}{Other}
        & Original & $3.06\times10^5$ & $5.54\times10^7$ & $5.14\times10^8$ & $4.86\times10^9$ & $1.10\times10^9$ 
        & \multirow{2}{*}{31} & \multirow{2}{*}{6/2021--12/2023} \\
        & Log     & 12.63 & 17.83 & 17.89 & 22.30 & 2.45
        & & \\
        \bottomrule
    \end{tabular}
    }
\end{table}

Based on the summary statistics presented in Table~\ref{tab:monthly_maximum_loss_long}, the distributions of monthly maximum losses across all three categories exhibit heavy right-skewness, indicative of the presence of extreme loss events. For the ETH category, the mean monthly maximum loss is $7.21\times10^7$, which is substantially higher than the median value of $2.74\times10^7$. The maximum observed loss of $6.02\times10^8$ and the standard deviation of $1.08\times10^8$ further suggest a heavy-tailed distribution. Similarly, the BSC category shows a mean monthly maximum loss of $6.50\times10^7$, which, although the lowest among the three categories, still  exceeds the median of $3.68\times10^6$. The maximum loss value of $1.30\times10^9$ and a high standard deviation of $2.24\times10^8$ confirm the presence of substantial outliers within the BSC loss data. For the Other category, the mean monthly maximum loss is $4.10\times10^8$, while the median stands at $3.60\times10^7$. The maximum recorded loss of $3.60\times10^9$ and a standard deviation of $8.45\times10^8$ further indicate significant dispersion and the influence of extreme values. For the monthly aggregate loss,  the mean of ETH is $1.05\times10^8$, which is substantially higher than the median value of $5.19\times10^7$. The maximum recorded loss reaches $6.96\times10^8$, with a standard deviation of $1.44\times10^8$.  In the BSC category, the mean is $7.59\times10^7$, the lowest among the three categories, yet still significantly exceeds the median of $1.01\times10^7$. The maximum loss observed is $1.40\times10^9$, and the standard deviation is $2.39\times10^8$.
For the Other category, the mean monthly aggregate loss amounts to $5.14\times10^8$, the highest across all categories, and also exceeds its corresponding median of $5.54\times10^7$. The maximum recorded loss is as high as $4.86\times10^9$, with a standard deviation of $1.10\times10^9$, indicating considerable dispersion and heavy-tailed characteristics.

To address the  skewness and stabilize the variability of the loss distributions, a log  transformation was applied to the loss data. It is observed from Table~\ref{tab:monthly_maximum_loss_long} that for both the maximum monthly loss and the aggregated monthly loss, the log transformed losses exhibit substantially reduced dispersion compared to the original scale.  Specifically, the standard deviations decrease significantly across all categories after the transformation. Moreover, the differences between the mean and median values are significantly narrower in the transformed losses, indicating a considerable reduction in the right skewness.



\begin{figure}[htb!]
    \centering
    \subfigure[ETH]{
        \includegraphics[width=.3\textwidth]{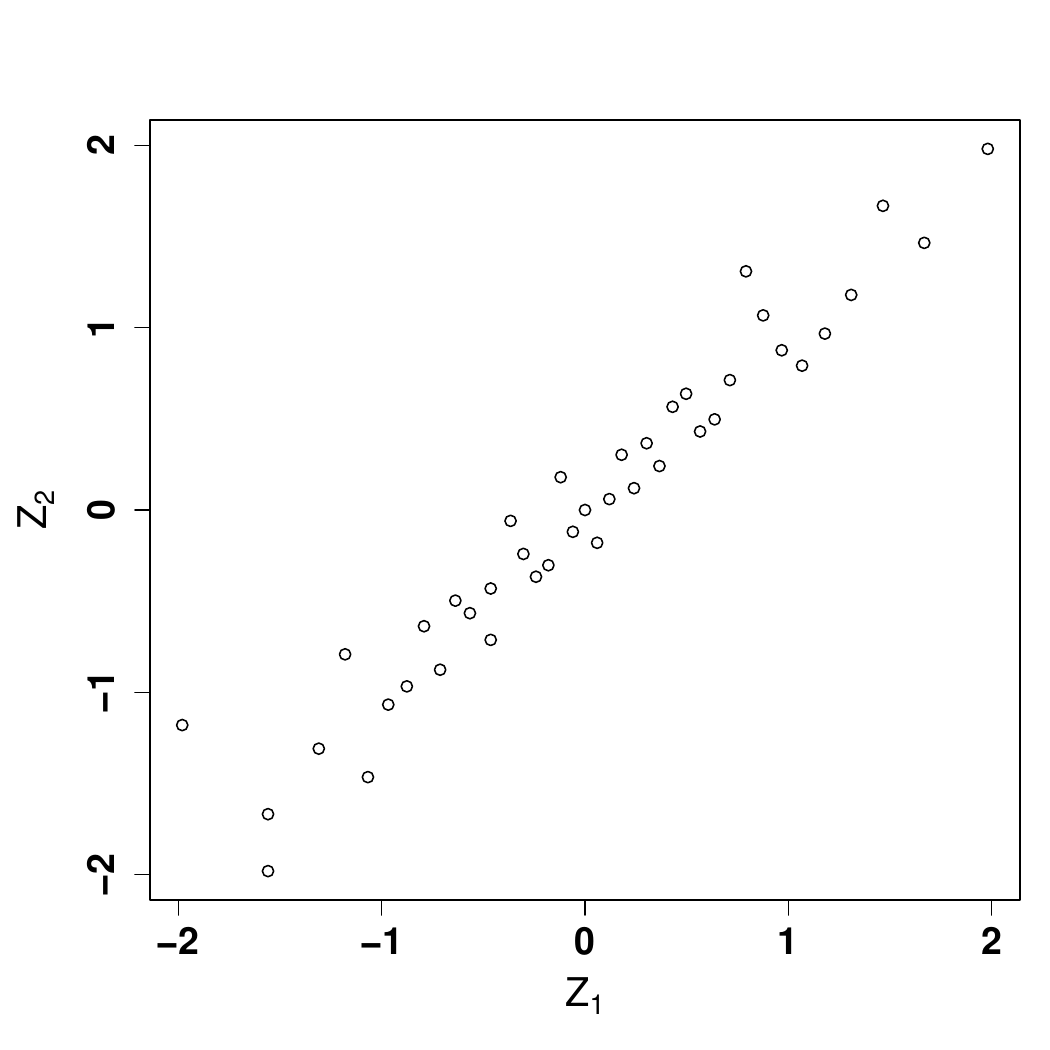}
        \label{fig:eth_normal_score}
    }
    \subfigure[BSC]{
        \includegraphics[width=.3\textwidth]{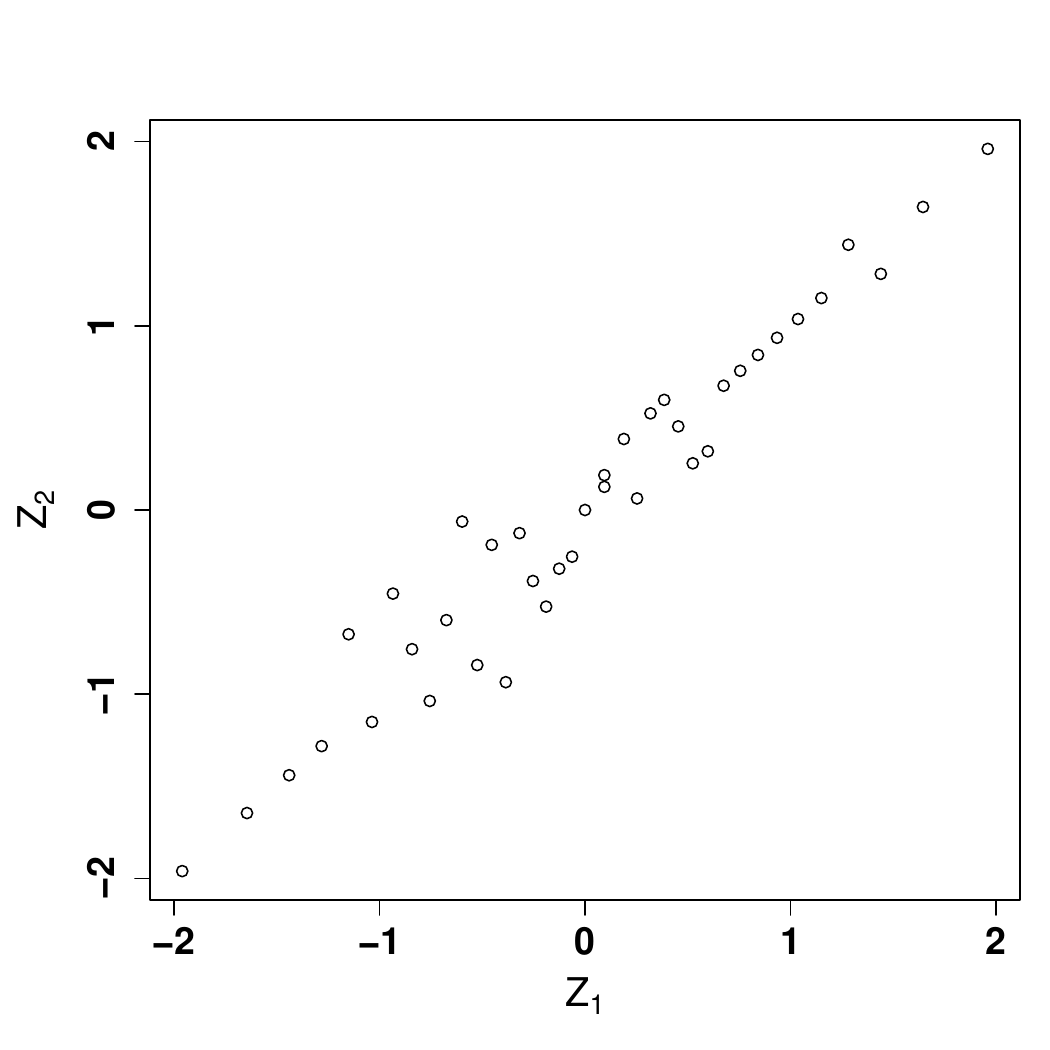}
        \label{fig:bsc_normal_score}
    }
    \subfigure[Other]{
        \includegraphics[width=.3\textwidth]{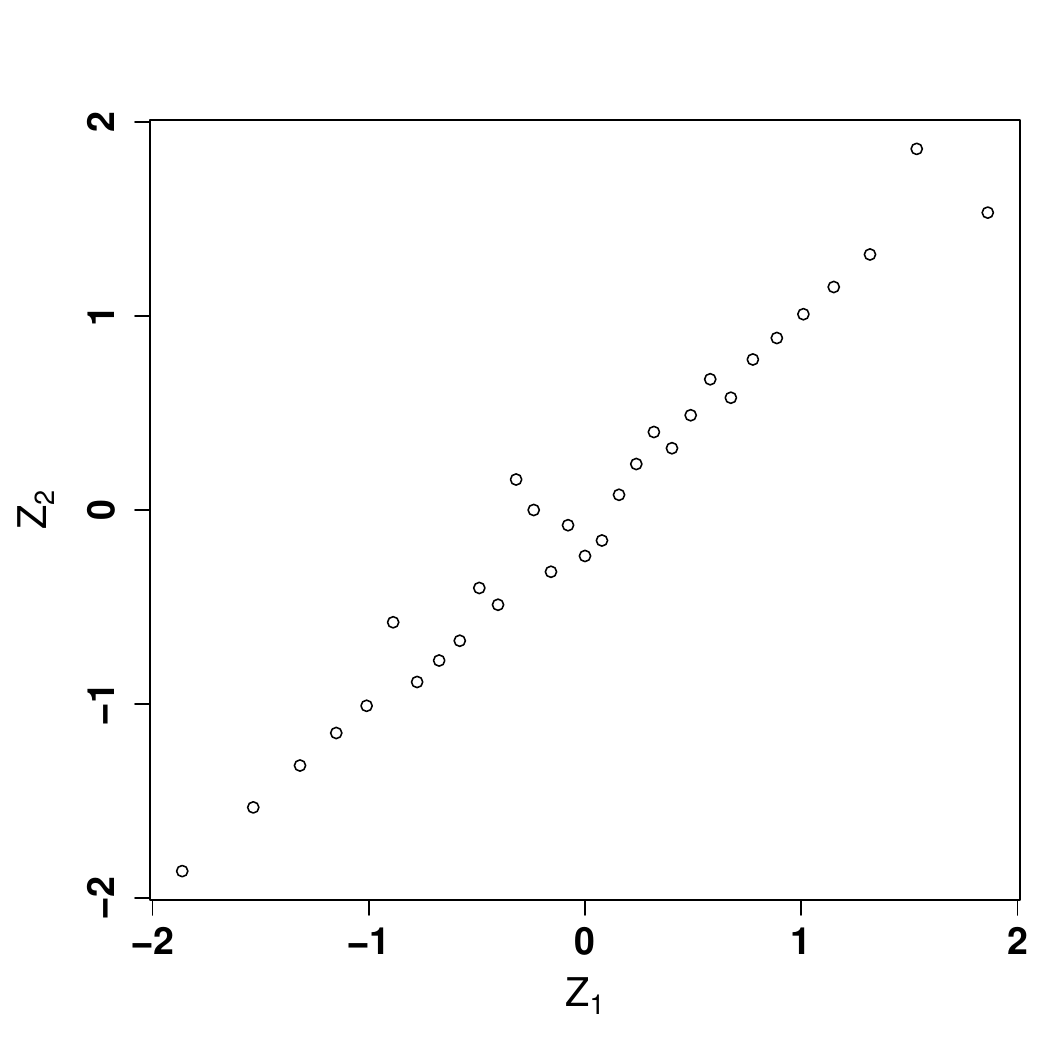}
        \label{fig:other_normal_score}
    }
    \caption{Normal score plots for all categories.\label{fig:normalscore}}
\end{figure}

To examine the dependence structure between monthly maximum losses and monthly aggregate losses, we present the normal score plots in Figure~\ref{fig:normalscore} \citep{joe2014dependence}. 
It is seen that all three categories exhibit strong positive dependence, as the points in each plot align closely along the diagonal line. In the ETH category Figure~\ref{fig:eth_normal_score}, the dependence appears uniform across the distribution, with no visible sign of tail dependence. The estimated Kendall’s \(\tau\) is 0.88, indicating strong association. For BSC and Other categories in Figures~\ref{fig:bsc_normal_score} and~\ref{fig:other_normal_score}, the dependence remains consistently strong, and Kendall’s \(\tau\) values are 0.87 and 0.92, respectively.   These findings suggest that the dependence between monthly maximum losses and monthly aggregate losses is substantial and should be explicitly accounted for in modeling across all categories.

\subsubsection{Financial Risks}
In our bond evaluation, we incorporate financial market risks from three sources: 
inflation,\footnote{\url{https://www.usInflationcalculator.com/Inflation/current-Inflation-rates/}} 
Treasury,\footnote{\url{https://fred.stlouisfed.org/series/DGS1}} 
and SOFR.\footnote{\url{https://fred.stlouisfed.org/series/SOFR}} 
The Treasury and inflation series span from January 2000 to December 2023, while SOFR is available from April 2018 through December 2023. 
All series are converted to monthly rates to ensure consistent temporal granularity and to facilitate subsequent joint modeling and analysis.  

Figure~\ref{fig:financial_ts_plot} displays the time series of these three risk factors.  

\begin{figure}[htb!]
    \centering
    \includegraphics[width=0.4\linewidth]{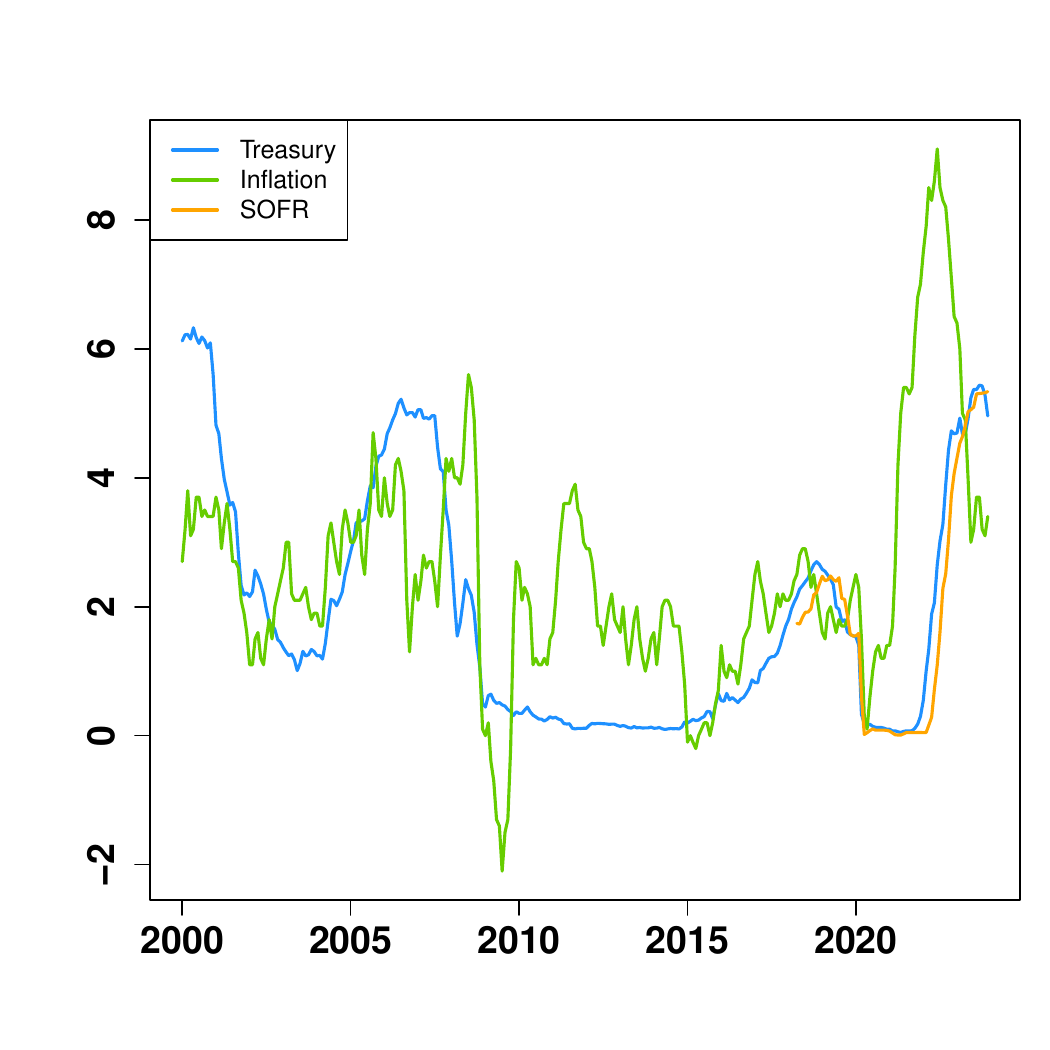}
    \caption{Time series of monthly financial risk factors: Treasury, inflation, and SOFR.}
    \label{fig:financial_ts_plot}
\end{figure}

All three financial series exhibit substantial volatility over time,   particularly around the 2008 financial crisis and the post-2020 inflation spike. This motivates the use of GARCH-type models to capture their time-varying conditional variances \citep{tsay2010analysis}; see Section 1.1 of the Supplementary Material.
Moreover, the plot reveals evidence of potential dependence among the series. For example, in the post-2020 period, all three rates rise sharply in response to macroeconomic conditions and monetary policy tightening.  These observations indicate that financial risks do not evolve independently but are driven by shared economic forces. From the perspective of crypto CAT bond pricing, this dependence is critical: bond cash flows are discounted jointly by these financial variables, so ignoring their correlation would misrepresent both expected values and tail behavior. To address this, we employ a multivariate time-series framework, with dependence captured through a copula construction (see Section 1.2 of the Supplementary Material), ensuring that correlated fluctuations in Treasury, inflation, and SOFR are coherently integrated into the distribution of bond prices.

\subsection{Estimation Procedure}\label{sec:estimation}
In this section, we present the estimation procedure.  As seen in Section \ref{sec:eda},  the monthly maximum loss \( Y_k \) and the monthly aggregate loss \( X_k \) are dependent. Therefore, the dependence between \( Y_k \)  and  \( Z_k \) should be accommodated. However, \( Y_k \) reflects localized extreme events and tends to exhibit heavy tails, while \( Z_k \) is an accumulated sum of losses and tends to exhibit smoother, more aggregated growth patterns. Therefore,  directly modeling \( (Y_k, Z_k) \) as a joint distribution would be overly restrictive and may fail to capture the distinct marginal features of each component. To address this, we instead focus on modeling the dependence between \(Y_k \) and $X_k$, which are observed on the same time scale and maintain a more interpretable relationship. This approach allows for greater modeling flexibility and facilitates the use of copula-based methods to capture complex dependencies between extreme and aggregate monthly losses.

To capture the dependence structure between the monthly maximum loss \( Y_k \) and the monthly aggregate loss \( X_k \), we adopt a \emph{copula} modeling approach \citep{joe2014dependence,shi2018pair,li2025ebicop}, which provides the flexibility to represent complex, potentially nonlinear dependencies beyond linear correlation. That is
\begin{equation}\label{eq:c1}
    P(Y_k\le y_k,X_k\le x_k)=C_1(F_k(y_k),G_k(x_k)),
\end{equation}
where \( C_1(\cdot) \) denotes the copula function capturing the dependence structure,  and \( F_k(y_k) \) and \( G_k(x_k) \) are the corresponding marginal distribution functions. The joint density function of \( Y_k \) and \( X_k \) is expressed as:
\[
f(y_k, x_k) = c_1(F_k(y_k), G_k(x_k)) \, f_k(y_k) \, g_k(x_k),
\]
where \( c_1(\cdot) \) denotes the copula density associated with $C_1(\cdot)$, and \( f_k(y_k) \) and \( g_k(x_k) \) are the marginal density functions of \( Y_k \) and \( X_k \), respectively. Given a sample of observed pairs \( \{(y_k, x_k)\}_{k=1}^T \), the log-likelihood function of the joint model is given by 
\begin{equation}\label{eq:crypo-ll}
  \ell(\bm{\theta}_1) 
= \sum_{k=1}^T  \log c_1\left( F_k(y_k), G_k(x_k)\right) + \sum_{k=1}^T \left[\log f_k(y_k) + \log g_k(x_k) \right],  
\end{equation}
where \( \bm{\theta}_1  \) represents the full parameter vector.

For the financial risk time series Treasury--\(\mathrm{TR}_k\), SOFR--\(\mathrm{R}_k\), and inflation rate--\(\mathrm{IF}_k\), as previously discussed, the data exhibit both high volatility and potential cross-series dependence. To model these characteristics, we adopt a two-stage approach:
\begin{itemize}
    \item We fit univariate ARIMA-GARCH models to each of the three series to capture autocorrelation, nonstationarity, and volatility clustering. The standardized residuals from these models are then transformed via their estimated distribution functions to yield uniform variables on the unit interval \([0,1]\).
    
    \item The uniform values are used as inputs to a multivariate copula model, which captures the residual dependence structure among the three series.
\end{itemize}
This two-stage framework separates marginal dynamics from cross-series dependence, allowing each component’s distinct time series behavior to be modeled while capturing the cross-series dependence \citep{nikoloulopoulos2012vine, peng2018modeling}. Let \({\bf W}_k = (W_{1,k}, W_{2,k}, W_{3,k})\) denote the vector of standardized residuals obtained from the fitted ARIMA-GARCH models for   \(\mathrm{TR}_k\), \(\mathrm{R}_k\), and    \(\mathrm{IF}_k\), respectively. We assume that \({\bf W}_k\) follows the joint distribution specified by a copula model:
\begin{equation}\label{copula-model}
\mathbb{P}({\bf W}_k \le {\bf w}_k) = C_2\left(\Psi_1(\text{w}_{1,k}), \Psi_2(\text{w}_{2,k}), \Psi_3(\text{w}_{3,k})\right),
\end{equation}
where \(\Psi_j\) is the marginal distribution function of \(W_{j,k}\) for \(j = 1,2,3\), and \(C_2\) is a trivariate copula function that captures the dependence structure among the standardized residuals. The corresponding joint log-likelihood function can be expressed as:
\begin{align*}
 \ell(\bm{\theta}_2)  = \sum_{k=1}^n \Big[ 
&\log c_2\left(\Psi_1(\text{w}_{1,k}), \Psi_2(\text{w}_{2,k}), \Psi_3(\text{w}_{3,k})\right) 
- \log(\sigma_{1,k}) - \log(\sigma_{2,k}) - \log(\sigma_{3,k}) \\
&+ \log \psi_1(\text{w}_{1,k}) + \log \psi_2(\text{w}_{2,k}) + \log \psi_3(\text{w}_{3,k}) 
\Big],
\end{align*}
where $\bm{\theta}_2$ is the parameter vector; \(c_2(\cdot)\) is the joint copula density associated with \(C_2(\cdot)\); \(\sigma_{j,t}\) is the conditional standard deviation of the \(j\)-th series at time \(k\) (see Eq. (2) of  the Supplementary Material); and {$\psi_ j(\cdot)$ denotes the marginal density function of $W_{j,k}$}, for \(j = 1, 2, 3\). 

We adopt the Inference Functions for Margins (IFM) method \citep{Joe1997} to estimate the parameters \((\bm{\theta}_1, \bm{\theta}_2)\). The IFM method provides a two-step semiparametric estimation procedure: in the first step, the parameters of the marginal models are estimated independently, and in the second step, the copula parameters governing the joint dependence structure are estimated conditional on the fitted marginals. The use of IFM is justified on both theoretical and practical grounds. 
Theoretically, it is consistent and asymptotically normal under standard conditions 
\citep{Joe1997,joe2014dependence}, and attains the same first-order efficiency as full maximum likelihood 
while being far less computationally demanding. Practically, IFM aligns with our model design: univariate time-series dynamics are handled at the marginal level, while copulas capture cross-series dependence. 
A full joint likelihood across all financial and crypto components would be computationally burdensome, whereas IFM avoids this complexity by decoupling marginal and dependence estimation. Moreover, 
in simulation-based bond pricing, where joint trigger probabilities drive cash flows, the ability to separately estimate and simulate from marginals and copulas is highly desirable. 

\subsection{Statistical Modeling of Crypto and Financial Risks}\label{sec:statistics}
This section presents the statistical framework for modeling crypto-related losses and financial risk factors. We specify marginal distributions, dependence structures, and time-series models that form the basis for bond pricing and simulation in subsequent sections.

\subsubsection{Modeling of Monthly Maximum and Aggregate Losses}

In this section, we describe the modeling framework applied to the log-transformed monthly maximum losses ($Y_k$) and monthly aggregate losses ($X_k$). According to classical extreme value theory  \citep{haan2006extreme}, the distribution of block maxima converges to the generalized extreme value (GEV) distribution:
\[
G(x) = \exp\left\{ -\left[1 + \xi \frac{x - \mu}{\sigma} \right]^{-1/\xi}_+ \right\},
\]
where $a_+ = \max\{a, 0\}$, $\mu \in \mathbb{R}$ is the location parameter, $\sigma > 0$ is the scale parameter, and $\xi \in \mathbb{R}$ is the shape parameter. To account for temporal variation in the loss process, the location and scale parameters may be modeled either as constants or as functions of time. We consider the following four modeling scenarios:

\begin{itemize}
    \item[$M_1$:] Both location and scale parameters are time-invariant:
    \(
    \mu(t) = \mu_0, \quad \sigma(t) = \sigma_0.
    \)

    \item[$M_2$:] The location parameter depends on time, while the scale parameter is constant:
    \(
    \mu(t) = \beta_0 + \beta_1 \log(t), \quad \sigma(t) = \sigma_0.
    \)

    \item[$M_3$:] The scale parameter varies linearly with time, while the location parameter is constant:
    \(
    \mu(t) = \beta_0, \quad \sigma(t) = \sigma_0 + \sigma_1 t.
    \)

    \item[$M_4$:] Both location and scale parameters are time-dependent:
    \(
    \mu(t) = \beta_0 + \beta_1 \log(t), \quad \sigma(t) = \sigma_0 + \sigma_1 t.
    \)
\end{itemize}


We focus on the GEV distribution for modeling both monthly maxima ($Y_k$) and monthly aggregate losses ($X_k$), given its flexibility and ability to capture extreme behavior.\footnote{To assess robustness, we also compare  the GEV fit with classical alternatives on the transformed data, including Normal, Gamma, and Weibull, with full results deferred to Section 2 of the Supplementary Material.} 




\begin{table}[htb!]
\centering
\caption{Comparison of transformed monthly  {maximum} loss modeling results for ETH, BSC, and Other categories. `Est.' denotes the parameter estimate; `SE' is the standard error. Stars indicate significance: *** $<0.01$; ** $<0.05$; * $<0.1$.}
\label{tab:combined_max_loss_models}
\resizebox{\textwidth}{!}{
\begin{tabular}{l|l|rr|rr|rr|rr}
\toprule
Category & Parameter & M1-Est. & M1-SE & M2-Est. & M2-SE & M3-Est. & M3-SE & M4-Est. & M4-SE \\
\midrule
\multirow{6}{*}{ETH} 
& $\beta_0$     & 16.3626*** & 0.3477 & 15.3871*** & 1.0763 & 15.3069*** & 1.1789 & 16.3833*** & 0.3706 \\
& $\beta_1$     & ---        & ---    & 0.3447     & 0.3624 & 0.3777     & 0.3929 & ---         & ---    \\
& $\sigma_0$  & 2.0573***  & 0.2676 & 2.0170***  & 0.2558 & 0.7996**   & 0.3338 & 0.7645**    & 0.2850 \\
& $\sigma_1$  & ---        & ---    & ---        & ---    & -0.0039    & 0.0120 & -0.0017     & 0.0100 \\
& $\xi$       & -0.4958*** & 0.0958 & -0.4729*** & 0.0897 & -0.5020*** & 0.1304 & -0.5113     & 0.1343 \\
& AIC         & {171.4153}   & ---    & 172.5443   & ---    & 174.4387   & ---    & 173.3865    & ---    \\
\midrule
\multirow{6}{*}{BSC} 
& $\beta_0$     & 14.7757*** & 0.3540 & 13.2818*** & 0.8105 & 14.3265*** & 1.6963 & 14.6074*** & 0.3303 \\
& $\beta_1$     & ---        & ---    & 0.5142*    & 0.2631 & 0.0998     & 0.5977 & ---         & ---    \\
& $\sigma_0$  & 1.9749***  & 0.2511 & 1.7697***  & 0.2568 & 0.9258**   & 0.4224 & 0.9823***   & 0.2395 \\
& $\sigma_1$  & ---        & ---    & ---        & ---    & -0.0171    & 0.0196 & -0.0198*    & 0.0105 \\
& $\xi$       & -0.1768    & 0.1135 & -0.0629    & 0.1670 & -0.1016    & 0.1379 & -0.1038     & 0.1337 \\
& AIC         & 174.4558   & ---    & 173.4231   & ---    & 174.7475   & ---    & {172.7748}   & ---    \\
\midrule
\multirow{6}{*}{Other} 
& $\beta_0$     & 16.8266*** & 0.5178 & 21.7279*** & 0.8931 & 21.1567*** & 0.8181 & 16.6845*** & 0.6534 \\
& $\beta_1$     & ---        & ---    & -2.0008*** & 0.3555 & -1.7521*** & 0.3583 & ---         & ---    \\
& $\sigma_0$  & 2.5261***  & 0.3943 & 1.6398***  & 0.2615 & 0.0807     & 0.3449 & 1.0612***   & 0.3460 \\
& $\sigma_1$  & ---        & ---    & ---        & ---    & 0.0225     & 0.0175 & -0.0063     & 0.0152 \\
& $\xi$       & -0.3728**  & 0.1634 & -0.0401    & 0.1785 & 0.0193     & 0.2020 & -0.4140**   & 0.1689 \\
& AIC         & 148.7564   & ---    & {135.4969}   & ---    & 135.9914   & ---    & 150.6172    & ---    \\
\bottomrule
\end{tabular}
}
\end{table}

Table~\ref{tab:combined_max_loss_models} presents the detailed estimation results and AIC values for the four GEV-based models fitted to the log-transformed monthly maximum losses across the ETH, BSC, and Other categories.  For ETH, model M1 yields the lowest AIC (171.42) among the four GEV models. Further, the QQ plot in  Figure~\ref{fig:ETH M1 QQ plot} demonstrates well alignment along the diagonal with minimal deviation across the distribution. 
For BSC, model M4 achieves the lowest AIC (172.77); however, visual inspection of its QQ plot ({see Section 2 of the Supplementary Material}) reveals  right-tail deviation. Model M2, with a slightly higher AIC (173.42), exhibits a more balanced fit across the distribution in  Figure~\ref{fig:BSC M2 QQ plot}. Thus,  M2 is selected for its better tail behavior and interpretability with a time-varying location parameter. For the Other category, model M2 clearly outperforms all other GEV models in terms of AIC (135.50 vs.\ 148.76 for M1). Its QQ plot in Figure~\ref{fig:Other M2 QQ plot} shows strong adherence to the diagonal and captures both central and extreme values well. All estimated parameters in M2 are also statistically significant at the 1\% level except for $\xi$, lending further support to its selection.

 \begin{figure}[htb!]
    \centering
    \subfigure[ETH-M1]{
        \includegraphics[width=.3\textwidth]{
        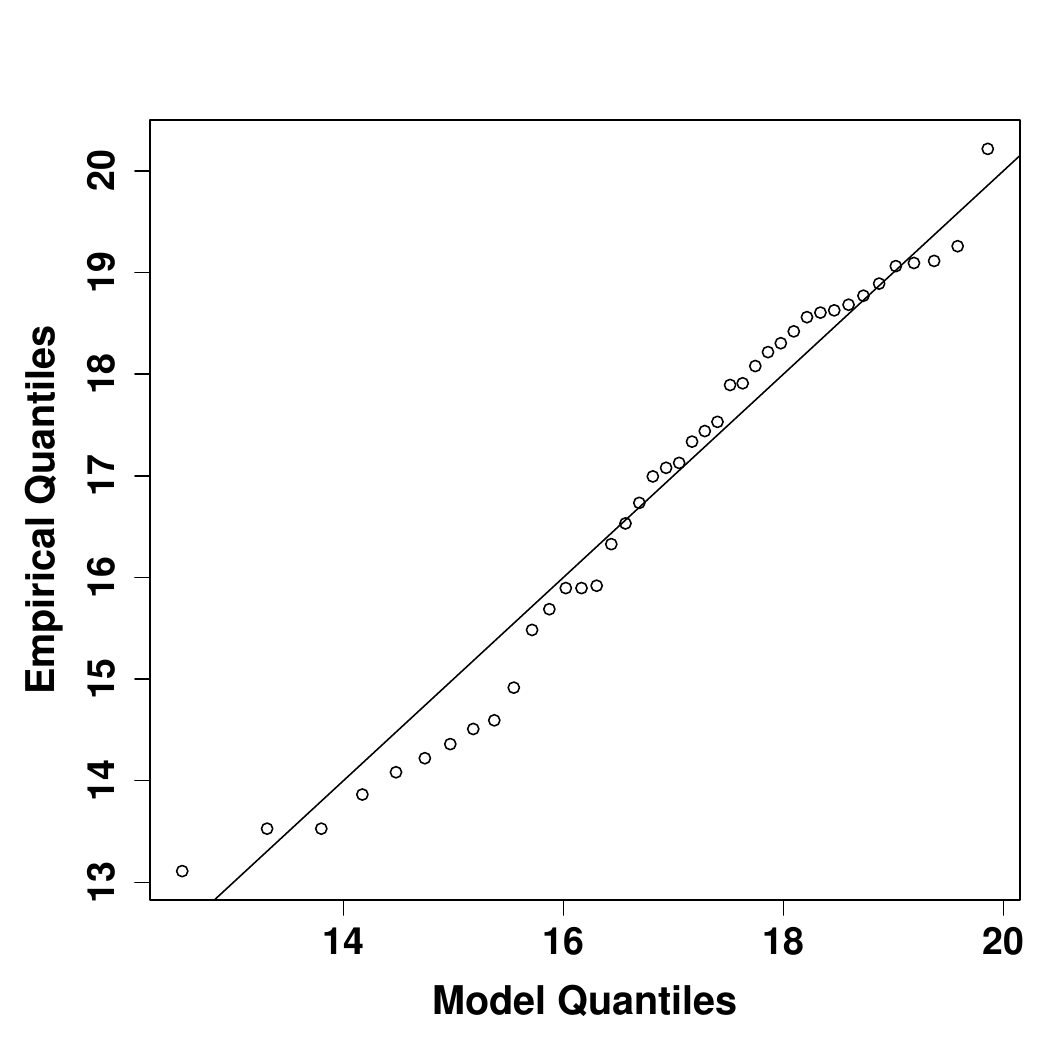}
        \label{fig:ETH M1 QQ plot}}
    \subfigure[BSC-M2]{
        \includegraphics[width=.3\textwidth]{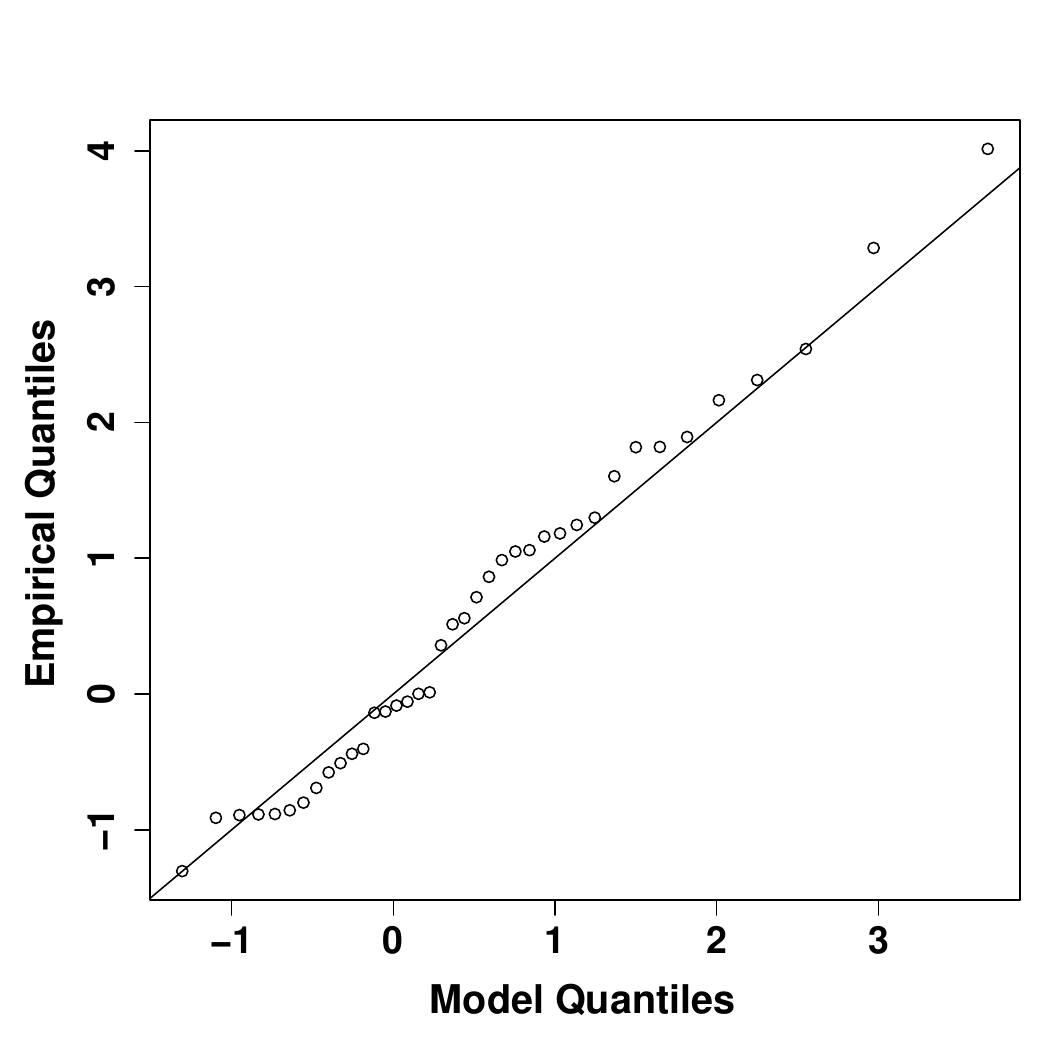}
        \label{fig:BSC M2 QQ plot}
    }
    \subfigure[Other-M2]{
        \includegraphics[width=.3\textwidth]{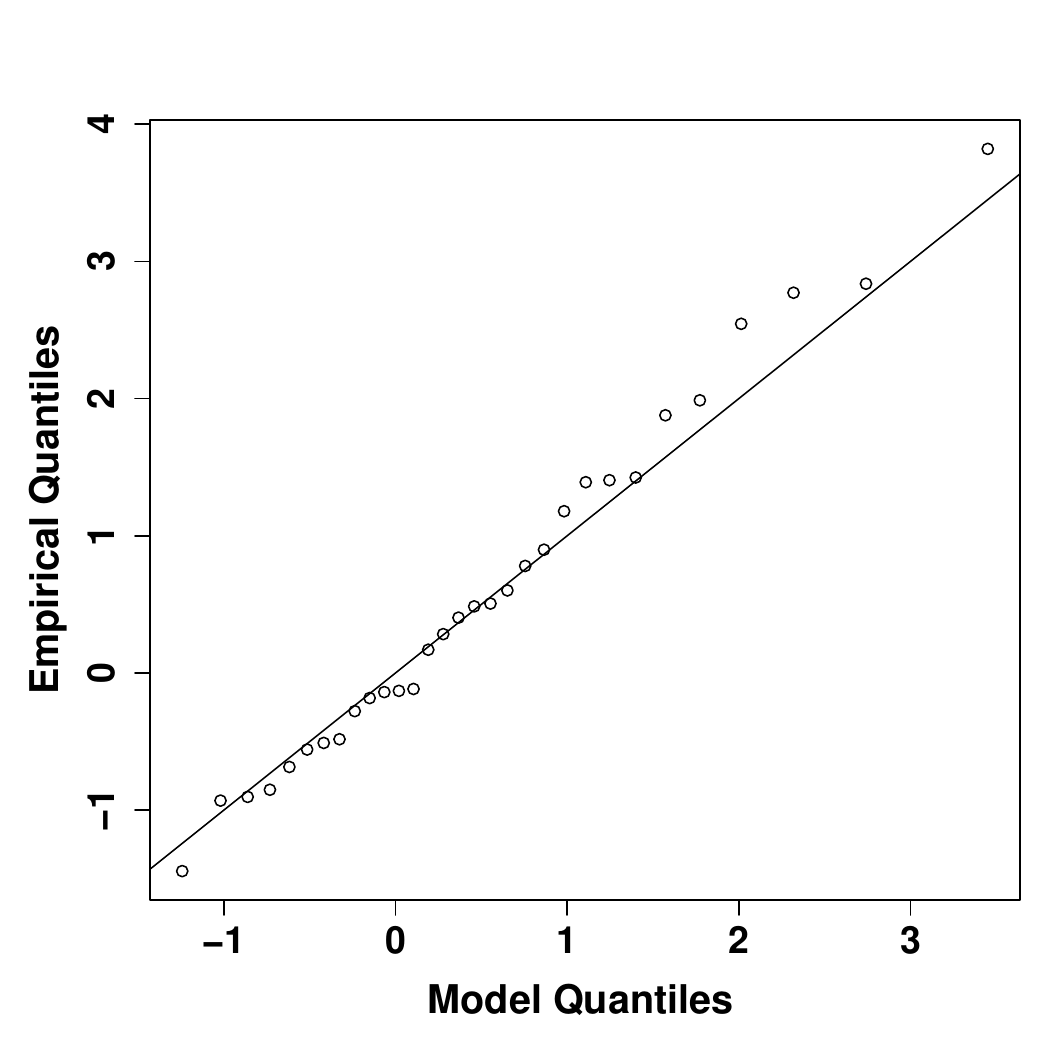}
        \label{fig:Other M2 QQ plot}
    }
    \caption{QQ plots for transformed monthly maximum losses of all categories.\label{fig:qqplot_month_maximum}}
\end{figure}

\begin{table}[htb!]
\centering
\caption{Comparison of monthly transformed  {aggregate} loss modeling results for ETH, BSC, and Other Categories. `Est.' denotes the parameter estimate; `SE' is the standard error. Stars indicate significance: *** $<0.01$; ** $<0.05$; * $p<0.1$.}
\label{tab:combined_sum_loss_models}
\resizebox{\textwidth}{!}{
\begin{tabular}{l|l|rr|rr|rr|rr}
\toprule
Category & Parameter & M1-Est. & M1-SE & M2-Est. & M2-SE & M3-Est. & M3-SE & M4-Est. & M4-SE \\
\midrule
\multirow{6}{*}{ETH} 
& $\beta_0$     & 17.0181*** & 0.3094 & 15.9685*** & 0.9377 & 15.9140*** & 1.0124 & 17.0181*** & 0.3245 \\
& $\beta_1$     & ---        & ---    & 0.3714     & 0.3160 & 0.3929     & 0.3404 & ---         & ---    \\
& $\sigma_0$  & 1.8191***  & 0.2393 & 1.7711***  & 0.2248 & 0.6347*    & 0.3183 & 0.5982**    & 0.2573 \\
& $\sigma_1$  & ---        & ---    & ---        & ---    & -0.0025    & 0.0115 & 0      & 0.0090 \\
& $\xi$       & -0.5003*** & 0.1040 & -0.4733*** & 0.0927 & -0.4913*** & 0.1258 & -0.5003***  & 0.1271 \\
& AIC         & {161.2360}   & ---    & 161.9099   & ---    & 163.8625   & ---    & 163.2360    & ---    \\
\midrule
\multirow{6}{*}{BSC} 
& $\beta_0$     & 15.5371*** & 0.3488 & 12.9501*** & 0.7444 & 14.3522*** & 1.5245 & 15.5021*** & 0.2637 \\
& $\beta_1$     & ---        & ---    & 0.9061***  & 0.2434 & 0.3783     & 0.5069 & ---         & ---    \\
& $\sigma_0$  & 2.0066***  & 0.2356 & 1.5853***  & 0.2144 & 0.9427**   & 0.3528 & 1.1268***   & 0.2211 \\
& $\sigma_1$  & ---        & ---    & ---        & ---    & -0.0226    & 0.0155 & -0.0309***  & 0.0095 \\
& $\xi$       & -0.2802*** & 0.0814 & -0.0838    & 0.1393 & -0.1628*** & 0.1067 & -0.1755*    & 0.0960 \\
& AIC         & 170.0244   & ---    & 163.7313   & ---    & 163.9808   & ---    & {162.4838}    & ---    \\
\midrule
\multirow{6}{*}{Other} 
& $\beta_0$     & 17.1663*** & 0.5047 & 21.6919*** & 0.9788 & 21.0939*** & 0.8410 & 17.2293*** & 0.6663 \\
& $\beta_1$     & ---        & ---    & -1.8464*** & 0.3857 & -1.5707*** & 0.3741 & ---         & ---    \\
& $\sigma_0$  & 2.5028***  & 0.3788 & 1.7461***  & 0.2663 & 0.1162     & 0.3359 & 0.8613*     & 0.4417 \\
& $\sigma_1$  & ---        & ---    & ---        & ---    & 0.0239     & 0.0170 & 0.0027      & 0.0196 \\
& $\xi$       & -0.3869**  & 0.1468 & -0.1117    & 0.1551 & -0.0572    & 0.1909 & -0.3726*    & 0.1948 \\
& AIC         & 147.5213   & ---    & {136.7736}  & ---    & 137.0258   & ---    & 149.5002    & ---    \\
\bottomrule
\end{tabular}
 }
\end{table}

Table~\ref{tab:combined_sum_loss_models} presents the parameter estimates and AIC values for GEV-based models M1 through M4, fitted to the log-transformed monthly aggregate losses across ETH, BSC, and Other categories.  For ETH, model M1 achieves the lowest AIC (161.24) among the GEV-based candidates and displays good alignment in the QQ plot of Figure~\ref{fig:ETH M1_sum QQ plot}. 
 For BSC, models M2 and M4 perform similarly in terms of AIC, with M4 slightly preferred (AIC 162.48 vs.\ 163.73 for M2). However, the QQ plot of M4 (see Section 2 of the Supplementary Material) shows noticeable tail discrepancies, while the QQ plot of M2 in Figure~\ref{fig:BSC M2_sum QQ plot} offers a more consistent fit across the distribution. As a result,  M2 is selected as the optimal model for BSC aggregate losses. {In the Other category, model M2 achieves the lowest AIC (136.77), substantially outperforming all other GEV models and classical alternatives. Its QQ plot in Figure~\ref{fig:Other M2_sum QQ plot} shows strong overall alignment with the theoretical quantiles, with only a minor deviation observed in the upper tail.
}

\begin{figure}[htb!]
    \centering
    \subfigure[ETH-M1]{
        \includegraphics[width=.3\textwidth]{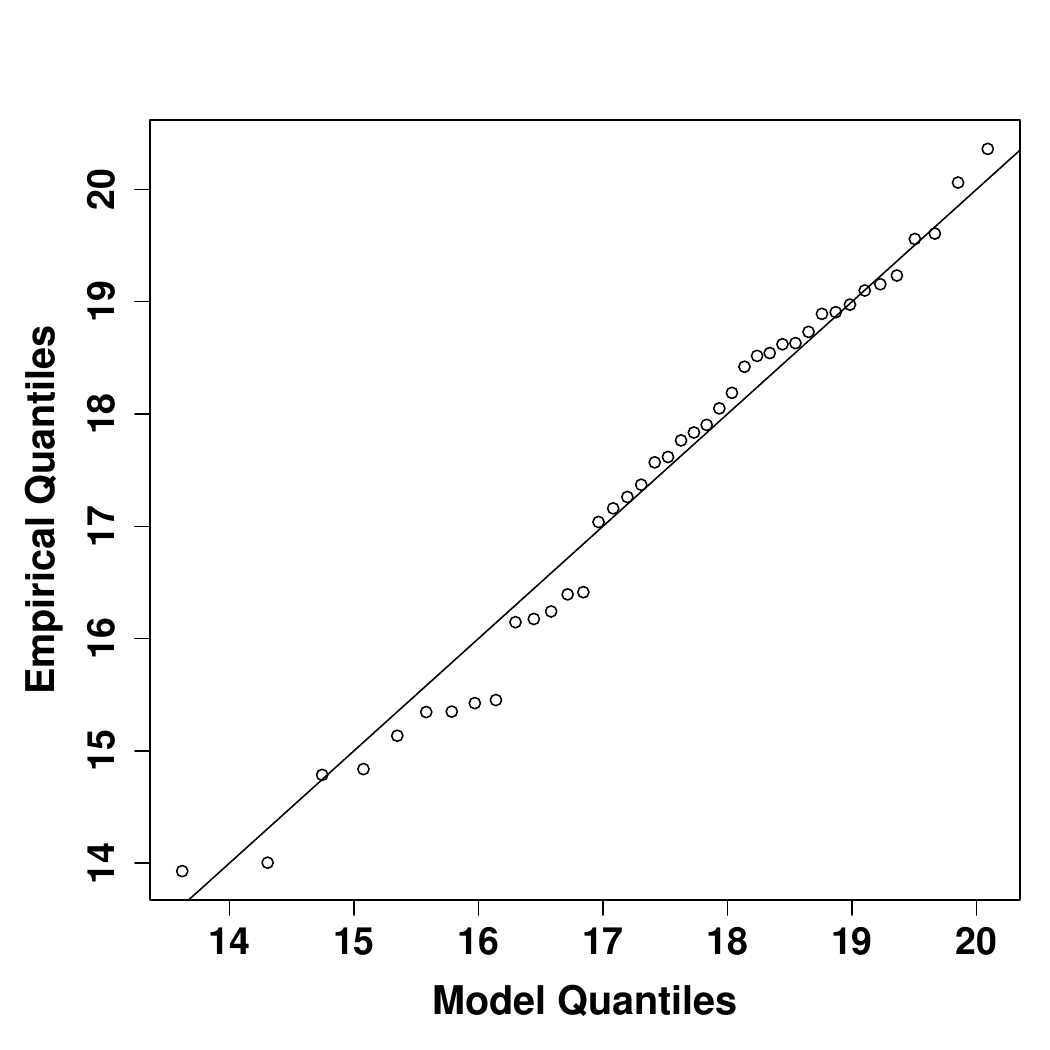}
        \label{fig:ETH M1_sum QQ plot}}
    \subfigure[BSC-M2]{
        \includegraphics[width=.3\textwidth]{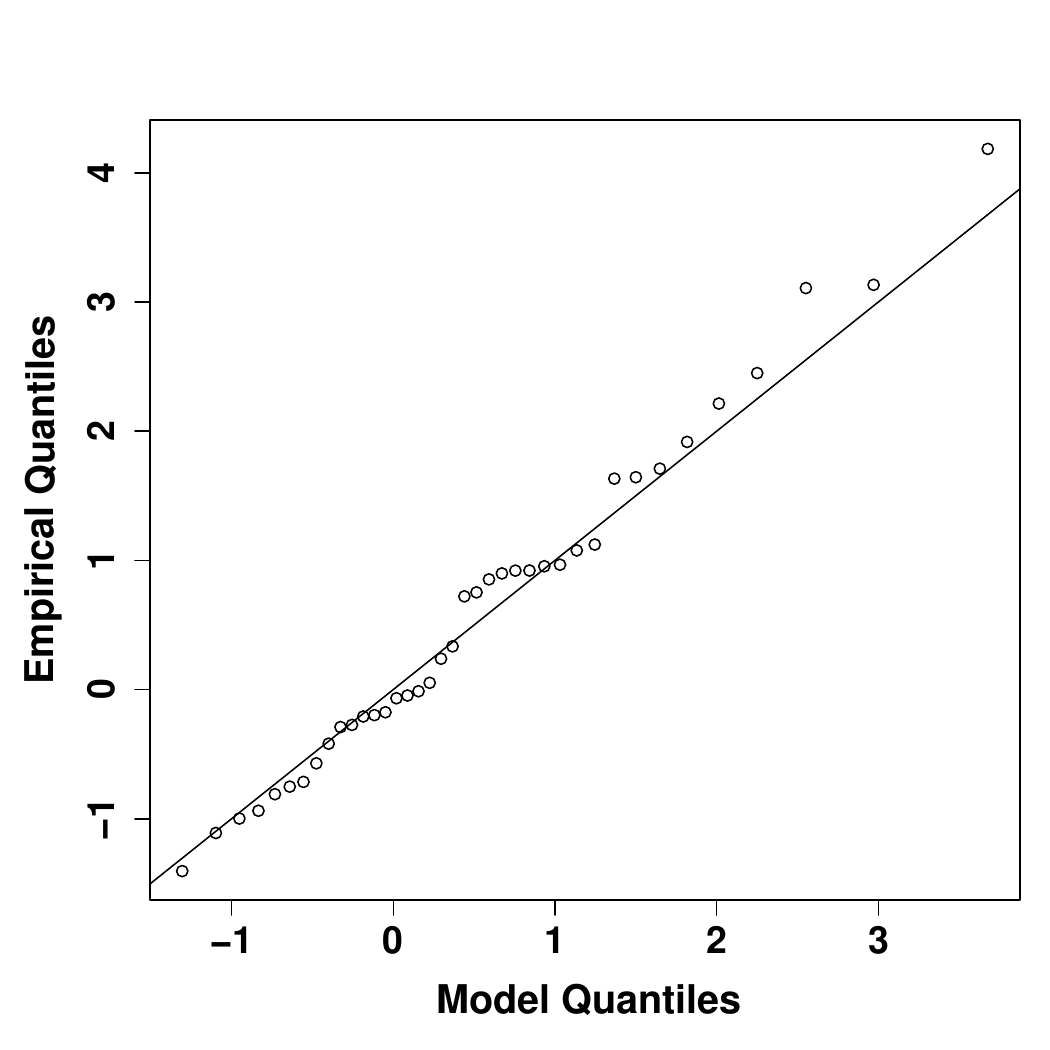}
        \label{fig:BSC M2_sum QQ plot}
    }
    \subfigure[Other-M2]{
        \includegraphics[width=.3\textwidth]{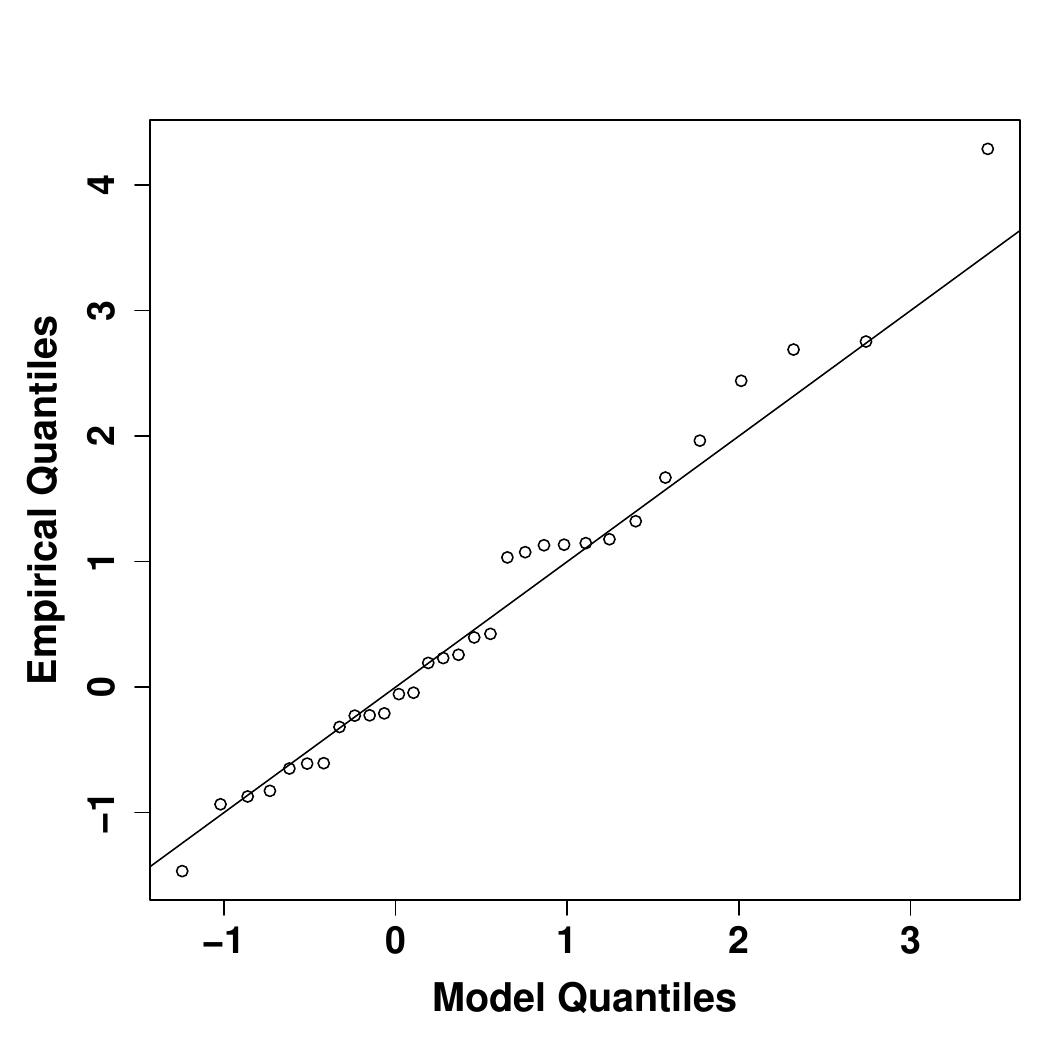}
        \label{fig:Other M2_sum QQ plot}
    }
    \caption{QQ plots for transformed monthly aggregate losses of all categories.\label{fig:qqplot_month_aggregate}}
\end{figure}

\subsubsection{Modeling the dependence between the joint
triggers}

To characterize the dependence between the monthly maximum losses and the monthly aggregate losses, we consider various bivariate copula structures \citep{nagler2022rvinecopulib}. It is quite interesting to observe that the best-fitting copula varies across categories as shown in Table \ref{tab:AIC}, highlighting heterogeneous dependence structures in different ecosystems. For the ETH category, the Gaussian copula achieves the lowest AIC. This elliptical dependence suggests that extreme monthly maxima in ETH are closely aligned with higher aggregate losses, but without distinct tail asymmetry. In contrast, the Gumbel copula performs best for BSC. The Gumbel family specifically emphasizes upper-tail dependence, indicating that extreme maximum losses on BSC are especially likely to coincide with elevated aggregate monthly losses.  For the Other category, the BB1 copula yields the lowest AIC. This copula captures both upper  and lower tail dependence, pointing to a more complex structure in smaller or emerging blockchains. Extreme maximum losses align with large aggregate outcomes. Low or moderate maxima also tend to coincide with relatively small aggregate losses.

\begin{table}[htbp]
\centering
\caption{AIC comparison of copula models for three data categories.}
\label{tab:AIC}
\resizebox{\textwidth}{!}{
\begin{tabular}{lcccccccccc}
\toprule
Type & Gaussian & T & Clayton & Gumbel & Frank & Joe & BB1 & BB6 & BB7 & BB8 \\
\midrule
ETH   & \textbf{-117.1474} & -115.3058 & -94.5175  & -112.9092 & -114.7765 & -94.4804  & -113.6769 & -110.8951 & -104.3914 & -91.8479 \\
BSC   & -119.1594 & -118.9340 & -108.7545 & \textbf{-121.3248} & -112.1043 & -108.7128 & -120.3079 & -118.4725 & -111.1547 & -97.9340 \\
Other & -117.7073  & -116.2552  & -105.4100  & -117.7305  & -103.1421  & -105.3604  & \textbf{-118.2391} & -113.2548 & -102.2504 & -83.6134 \\
\bottomrule
\end{tabular}
}
\end{table}
Table~\ref{tab:loss copula prams} reports the estimated copula parameters and corresponding Kendall’s $\tau$ for the best-fitting copulas identified in Table~\ref{tab:AIC}. It is observed that all the estimated parameters are significant with large $\tau$ values. 

\begin{table}[htbp]
\centering
\caption{Copula parameter estimates and Kendall’s $\tau$ for different categories' loss. Stars indicate significance: *** $<0.01$; ** $<0.05$; * $<0.1$.}
\label{tab:loss copula prams}
\begin{tabular}{lcccc}
\toprule
Category & Copula Family & Par1 & Par2 & Kendall’s $\tau$ \\
\midrule
ETH   & Gaussian & $0.97^{***}$ & --   & 0.85 \\ 
BSC   & Gumbel   & $7.32^{***}$ & --   & 0.86 \\ 
Other & BB1      & $1.51^{**}$ & $5.97^{***}$ & 0.90 \\ 
\bottomrule
\end{tabular}
\end{table}

To further visualize the dependence structures, Figure~\ref{fig:contour} presents the selected contour plots, where the blue points represent the normal scores obtained via the probability integral transform of the marginals. For the ETH category in Figure~\ref{fig:ETH fitted contour}, the elliptical contours align well with the point cloud, supporting the earlier AIC result that favored the Gaussian copula. In the BSC plot in Figure~\ref{fig:BSC fitted contour}, the contours exhibit noticeable asymmetry with heavier density toward the upper-right quadrant, consistent with the Gumbel copula's upper-tail dependence. In the Other category in Figure~\ref{fig:Other fitted contour}, the contour shape is more flexible and accommodates both tails, aligning with the BB1 copula selection, which captures more complex tail dependencies. Overall, these visualizations confirm the model selection results. They also show the need for category-specific copulas to capture distinct dependence patterns in cyber loss data across blockchain ecosystems.

\begin{figure}[htb!]
    \centering
    \subfigure[ETH]{
        \includegraphics[width=.3\textwidth]{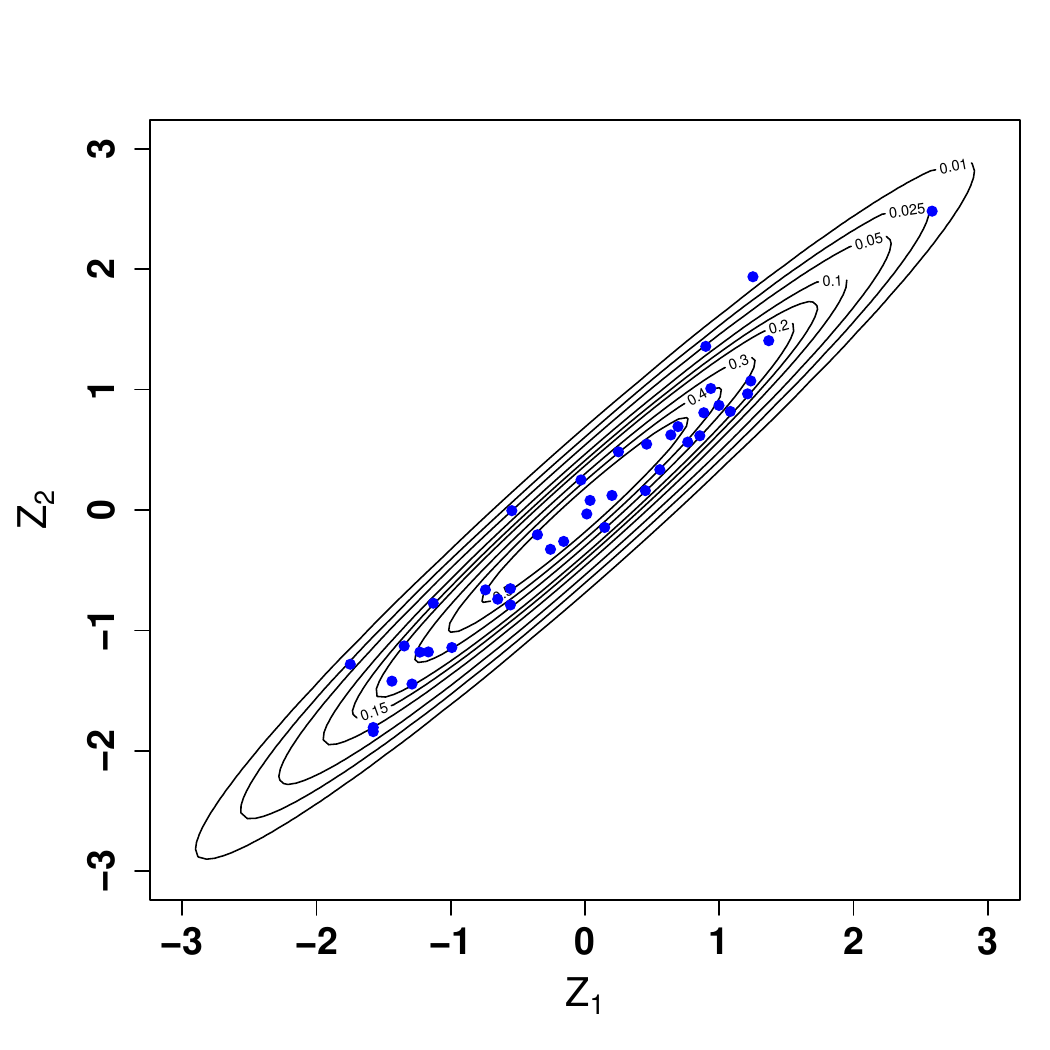}
        \label{fig:ETH fitted contour}}
    \subfigure[BSC]{
        \includegraphics[width=.3\textwidth]{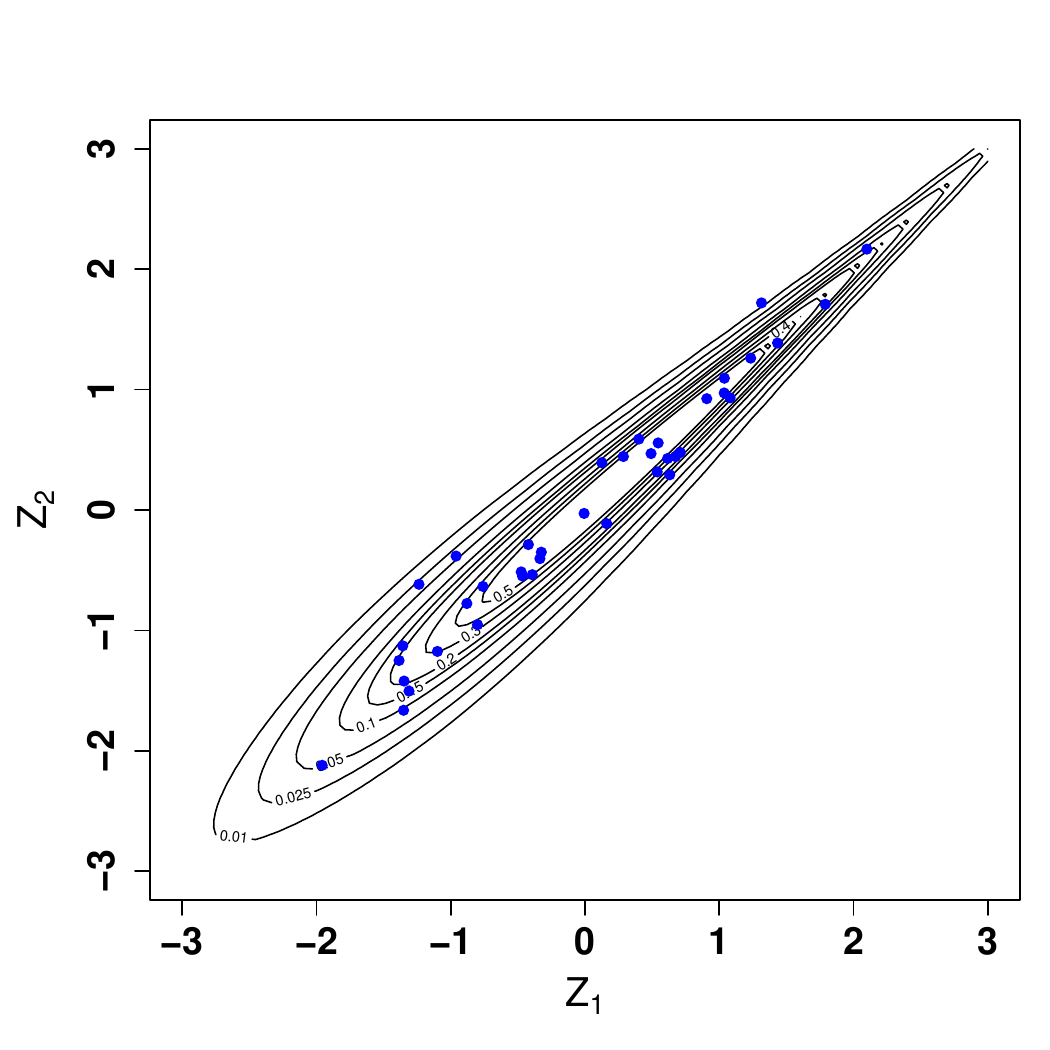}
        \label{fig:BSC fitted contour}
    }
    \subfigure[Other]{
        \includegraphics[width=.3\textwidth]{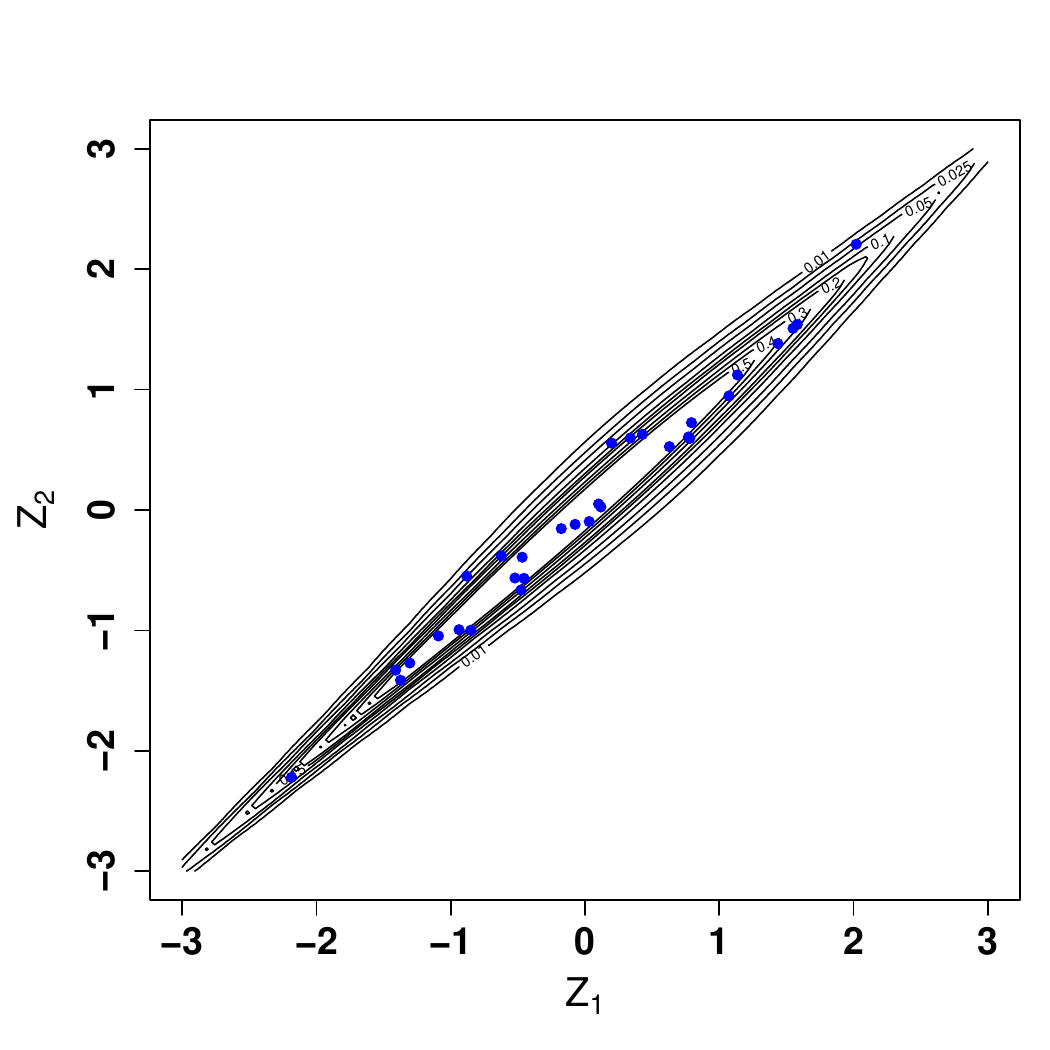}
        \label{fig:Other fitted contour}
    }
    \caption{Contour plots of all blockchain categories, where the blue points are  normal scores.\label{fig:contour}}
\end{figure}

\subsubsection{Modeling financial risks}

\paragraph{Marginal fitting.}
For the financial risk data, we first model each time series individually using ARIMA-GARCH specifications. Model selection for the ARIMA\((p,d,q)\) structure is based on the AIC, with the autoregressive order \(p\) and moving average order \(q\) allowed to vary between 0 and 5. When multiple models yield comparable AIC values, we favor the more parsimonious specification to avoid overfitting. A well-specified ARIMA-GARCH model should capture both the serial correlation in the conditional mean and the volatility clustering in the conditional variance. Accordingly, if the model is adequate, no significant autocorrelation should remain in either the standardized residuals or their squared values. Diagnostic checks on the residuals suggest that a GARCH(1,1) structure is sufficient to capture the volatility dynamics for all series considered. This finding is consistent with previous studies indicating that higher-order GARCH models often do not offer substantial improvements over GARCH(1,1) in empirical applications \citep{peng2018modeling, nikoloulopoulos2012vine}. To account for potential non-normality and asymmetry in the innovation distribution, we consider a range of candidate error distributions for the GARCH innovations \( W_t \), including the normal (NORM), generalized error distribution (GED), Student's \(t\) (STD), skewed Student's \(t\) (SSTD), skewed generalized error distribution (SGED), and skewed normal (SNORM) families \citep{ghalanos2020introduction}. This flexible error specification enhances the robustness of our marginal models to heavy tails and skewness, which are commonly observed in financial time series.

\begin{table}[htb!]
\centering
\caption{Parameter estimates for ARIMA-GARCH models fitted to the Treasury, inflation, and SOFR. `EST.' represents the estimates, and `SE' represents the standard errors. Stars indicate significance: *** $<0.01$; ** $<0.05$; * $<0.1$.}
\begin{tabular}{lccccccccc}
\toprule
 & \multicolumn{3}{c}{Treasury} & \multicolumn{3}{c}{Inflation} & \multicolumn{3}{c}{SOFR} \\
\cline{2-4} \cline{5-7} \cline{8-10}
Para. & EST. & SE & & EST. & SE & & EST. & SE & \\
\midrule
$\phi_1$ & 0.5283$^{***}$ & 0.0176 & & -0.0486$^{***}$ & 0.0003 & & -0.7701$^{***}$ & 0.0531 & \\
$\phi_2$ & 0.3266$^{***}$ & 0.0157 & & -0.3751$^{***}$ & 0.0002 & & 0.6673$^{***}$ & 0.0357 & \\
$\phi_3$ & - & - & & 0.5057$^{***}$ & 0.0002 & & 0.6079$^{***}$ & 0.0486 & \\
$\theta_1$ & -0.1439$^{***}$ & 0.0212 & & 0.6061$^{***}$ & 0.0001 & & 1.1632$^{***}$ & 0.0019 & \\
$\theta_2$ & -0.4094$^{***}$ & 0.0303 & & 0.6612$^{***}$ & 0.0001 & & -0.4079$^{***}$ & 0.0018 & \\
$\theta_3$ & - & - & & -0.4104$^{***}$ & 0.0001 & & -0.5952$^{***}$ & 0.0017 & \\
$\omega$ & 0.0003$^{***}$ & 0.0000 & & 0.0179 & 0.0227 & & 0.0014 & 0.0012 & \\
$\alpha_1$ & 0.3920$^{***}$ & 0.0265 & & 0.2011 & 0.1545 & & 0.8049$^{***}$ & 0.1854 & \\
$\beta_1$ & 0.6070$^{***}$ & 0.0324 & & 0.6789$^{**}$ & 0.2932 & & 0.1940$^{*}$ & 0.1036 & \\
$\text{skew}$ & 0.8996$^{***}$ & 0.0297 & & 0.8730$^{***}$ & 0.0792 & & 0.8840$^{***}$ & 0.1064 & \\
$\text{shape}$ & 0.9156$^{***}$ & 0.0321 & & 1.8530$^{***}$ & 0.2800 & & - & - & \\ \hline
$p$ & 2&- & &3&-& &3&-\\
$d$ & 1 &-& &1&-& &1&-\\
$q$& 2&-& &3&-& &3&-\\ \hline
$\text{LB: z}$ & 0.1798 & - & & 0.7001 & - & & 0.6015 & - & \\
$\text{LB: z}^2$ & 0.8302 & - & & 0.1049 & - & & 0.7996 & - & \\
\bottomrule
\end{tabular}
\label{tab:financial fitting information}
\end{table}

{Table~\ref{tab:financial fitting information} presents the model selection results and parameter estimates for the Treasury, inflation, and SOFR series. For the Treasury, the selected model is ARIMA(2,1,2) with GARCH(1,1) errors and SGED  innovations. All estimated parameters are statistically significant at the 0.01 level. The Ljung--Box (LB) test \(p\)-values for the standardized residuals and their squares are 0.1798 and 0.8302, respectively, suggesting that the model has adequately captured both the linear and volatility dynamics of the Treasury series. For the Inflation, the best-fitting model is ARIMA(3,1,3) + GARCH(1,1), also with SGED-distributed innovations. Most parameters are statistically significant, except for the GARCH parameters \(\omega\) and \(\alpha_1\). The LB test \(p\)-values for the standardized residuals and squared residuals are 0.7001 and 0.1049, respectively, indicating no significant autocorrelation remains and suggesting that the model provides a reasonable fit to the data. For SOFR, we select an ARIMA(3,1,3) + GARCH(1,1) model with SNORM innovation distribution. Except for $\omega$, all other estimated parameters are statistically significant. The LB test \(p\)-values for the standardized and squared residuals are 0.6015 and 0.7996, respectively, indicating that the model successfully accounts for both autocorrelation and conditional heteroskedasticity in the SOFR series.}

\begin{figure}[htb!]
    \centering
    \subfigure[Treasury Rate]{
        \includegraphics[width=.3\textwidth]{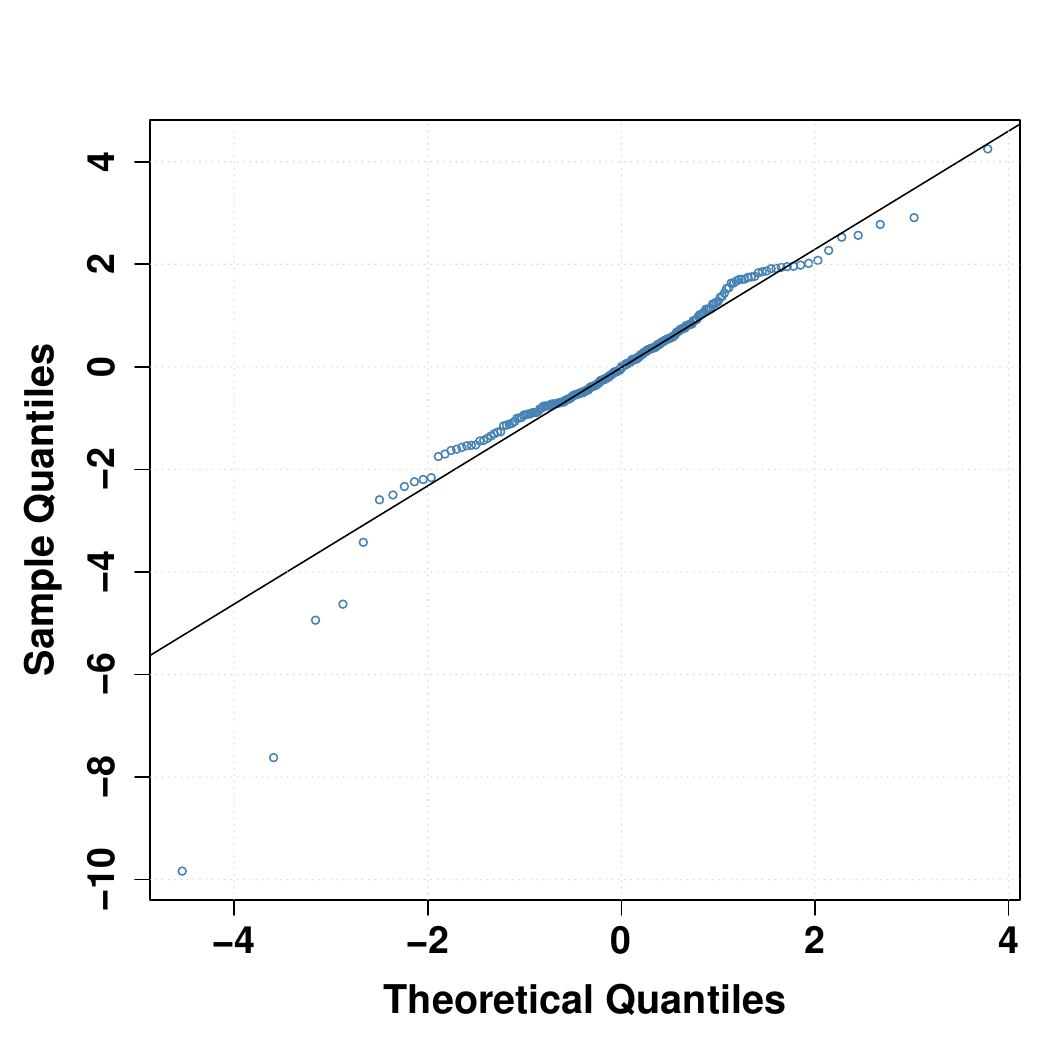}
        \label{fig:Treasury qq plot sged}
    }
    \subfigure[Inflation Rate]{
        \includegraphics[width=.3\textwidth]{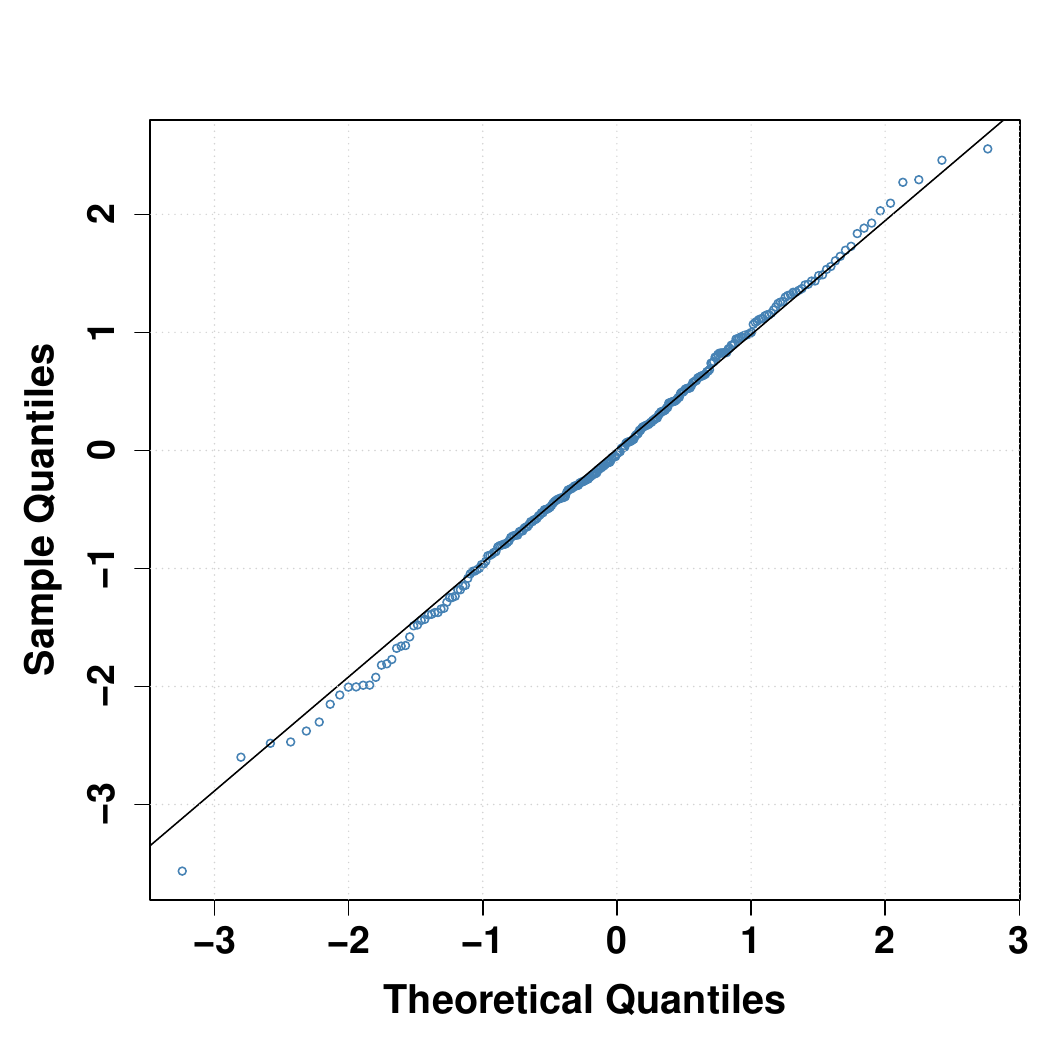}
        \label{fig:Inflation qq plot sged}
    }
    \subfigure[SOFR]{
        \includegraphics[width=.3\textwidth]{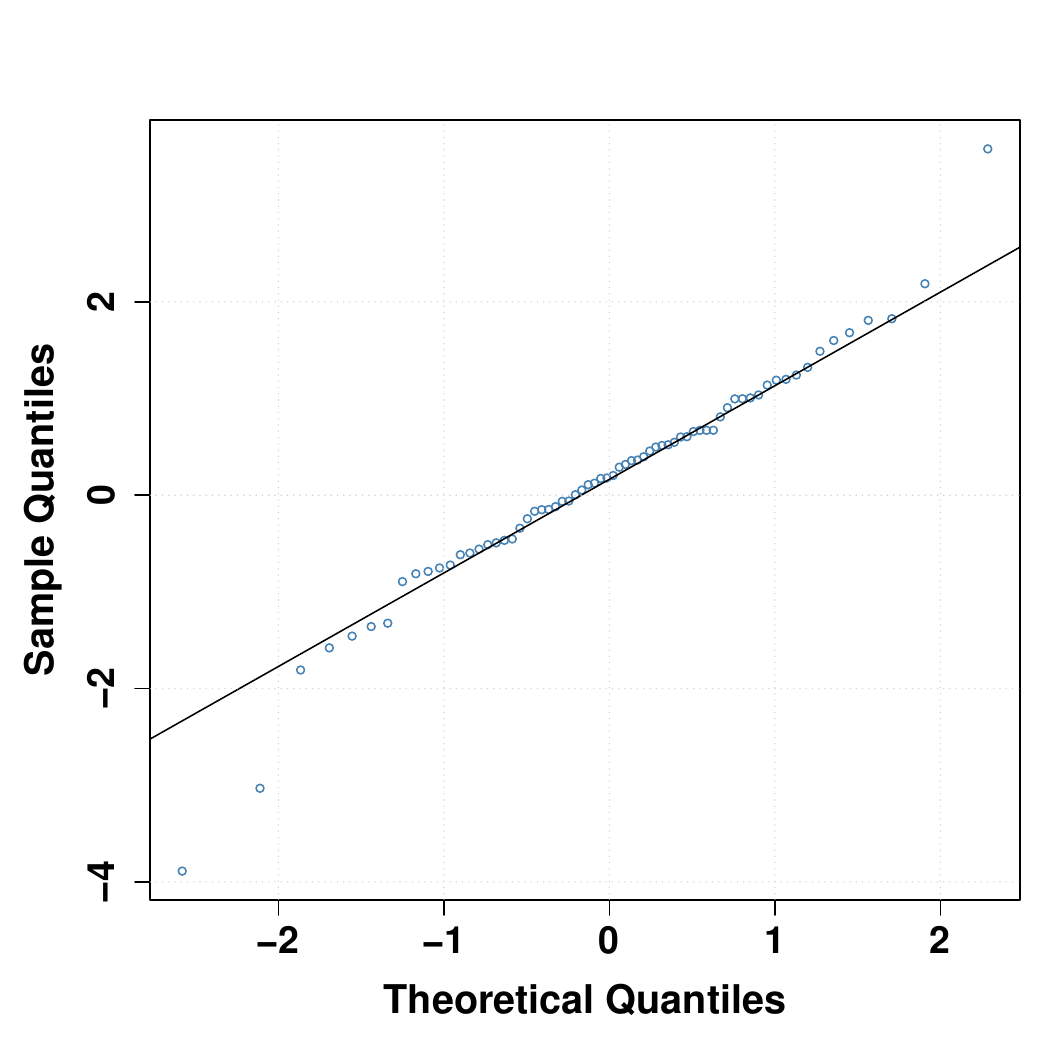}
        \label{fig:sofr qq plot snorm}
    }
    \caption{QQ plots of fitted financial risk rates.}
    \label{fig:financial qq plots}
\end{figure}

To further assess the goodness-of-fit of the proposed ARIMA-GARCH models, Figure~\ref{fig:financial qq plots} presents QQ plots of the standardized residuals for the Treasury, inflation, and SOFR.   Figure \ref{fig:Treasury qq plot sged} indicates that the Treasury residuals follow the theoretical distribution well in the central region, although some deviations are observed in the lower tail, suggesting occasional extreme negative shocks not fully captured by the model.  Figure \ref{fig:Inflation qq plot sged} shows that the inflation data aligns closely with the theoretical quantiles across almost the entire distribution. This confirms that the SGED-based model captures the empirical distribution well, including asymmetries and heavy tails. Figure \ref{fig:sofr qq plot snorm}, corresponding to SOFR, also shows close adherence to the theoretical line, especially in the central region, with minor deviations in both tails. This suggests that the model captures the bulk of the distribution well, with only limited departure in extreme observations. In addition, the Kolmogorov–Smirnov test is conducted for each residual series to evaluate the adequacy of the marginal distribution fitting. The resulting $p$-values are 0.4891 (Treasury), 0.9995 (inflation), and 0.3503 (SOFR), respectively.  

Overall, the QQ plots provide visual confirmation that the residuals from the fitted ARIMA-GARCH models conform reasonably well to the assumed innovation distributions.

\paragraph{Dependence modeling.}
To investigate the dependence structure among the financial risk factors, we compute both Pearson’s correlation coefficients ($\rho$) and Kendall’s $\tau$  based on the standardized residuals from the fitted ARIMA-GARCH models for the Treasury, inflation, and SOFR. To examine the cross-series dependence in a consistent manner, we extract the standardized residuals from the marginal models and use temporally aligned samples of size 68, corresponding to the shorter data availability for SOFR. Although the Treasury and inflation series span a longer period, the alignment ensures valid joint modeling across all three financial time series.

\begin{figure}
    \centering
    \includegraphics[width=0.4\linewidth]{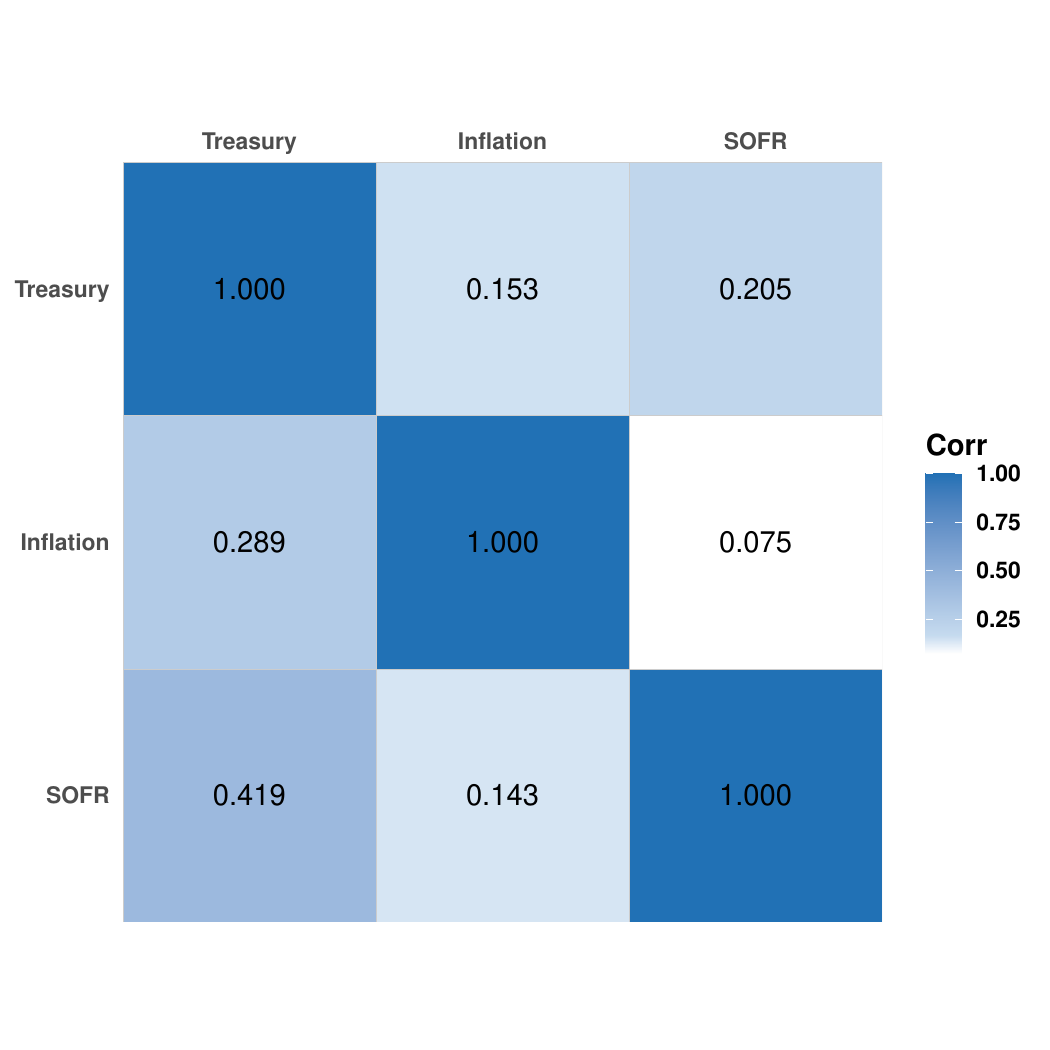}
    \caption{Heatmap of Pearson’s $\rho$ (lower triangle) and Kendall’s $\tau$ (upper triangle) correlations among the standardized residuals of Treasury, inflation, and SOFR time series. }
    \label{fig:contour pobs}
\end{figure}

As shown in Figure~\ref{fig:contour pobs}, all variable pairs exhibit moderate positive dependence, though the strength varies across pairs. The strongest dependence is observed between Treasury and SOFR, with $\rho = 0.419$ and $\tau = 0.205$. The correlation between Treasury and inflation is also non-negligible ($\rho = 0.2889$, $\tau = 0.153$), while the weakest dependence is found between inflation and SOFR ($\rho = 0.143$, $\tau = 0.075$). These asymmetric correlation patterns highlight the importance of a flexible dependence modeling framework that can capture pairwise and potentially higher-order conditional dependencies. To accommodate these heterogeneous relationships, we employ a three-dimensional C-vine copula model. The C-vine structure is particularly well suited for small-dimensional dependence modeling where one variable, such as the Treasury rate, can act as a central node, conditioning the dependence between the other pairs. This approach not only allows for distinct copula families for each pairwise link, but also enables the modeling of conditional dependence in a tractable yet expressive manner.

\begin{table}[htb!]
\centering
\caption{Summary of fitted C-vine copula structure for standardized residuals. Treasury is selected as the root node (node 1) according to the AIC criterion, followed by Inflation (node 2) and SOFR (node 3). Stars indicate significance: *** $<0.01$; ** $<0.05$; * $<0.1$.}
\label{tab:psnorm}
\begin{tabular}{@{}cllll@{}}
\toprule
\textbf{Tree} & \textbf{Edge} & \textbf{Copula Family} & \textbf{Parameter} & \textbf{Kendall’s $\bm{\tau}$} \\
\midrule
1 & (1, 2)      & Clayton        & $0.15^{*}$    & 0.07 \\
1 & (3, 1)      & Clayton        & $0.29^{***}$   & 0.13 \\
2 & (3, 2 \textbar 1) & Independence   & --      & 0.00 \\
\bottomrule
\end{tabular}
\end{table}
As discussed in the previous section, the standardized residuals from the ARIMA-GARCH models for the three financial time series are modeled using different innovation distributions. Specifically, the  SGED  is employed for the Treasury and inflation residuals, while the  SNORM  is used for SOFR. The standardized residuals are transformed into the unit interval $[0,1]$ using those fitted distributions.

\begin{figure}[htb!]
    \centering
    \includegraphics[width=0.35\linewidth]{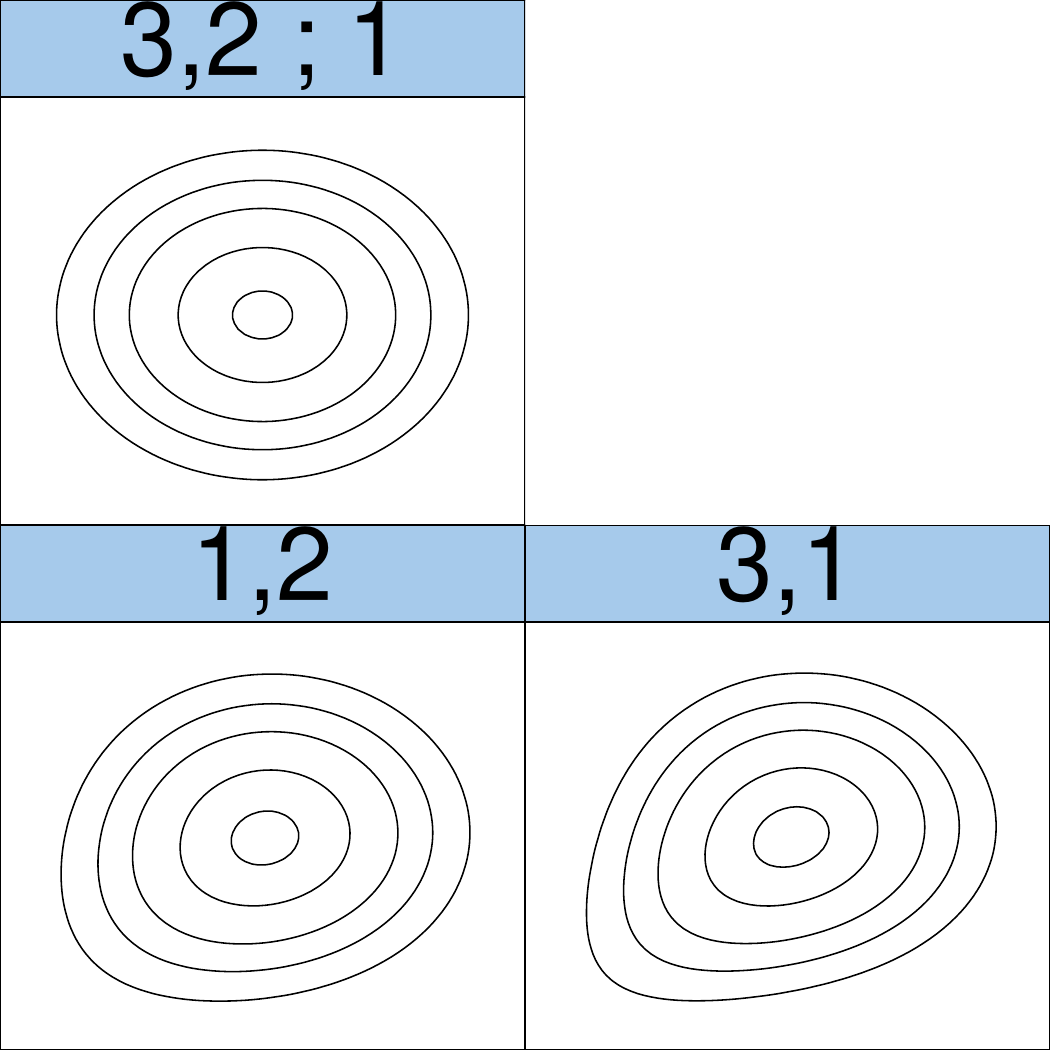}
    \caption{Contour plots illustrating the fitted C-vine copula structure linking Treasury, inflation, and SOFR residuals.}
    \label{fig:contour_psnorm}
\end{figure}

For the vine copula construction, as shown in Table \ref{tab:psnorm}, Treasury is selected as the root node (node 1) based on the AIC, with 
inflation and SOFR assigned as nodes 2 and 3, respectively. In the first tree, the edge between Treasury 
and inflation is modeled by a Clayton copula, with an estimated parameter of 0.15 and a Kendall’s $\tau$ 
of 0.07, indicating weak lower-tail dependence and low overall concordance. The edge between Treasury and SOFR is modeled with a Clayton copula, giving a parameter of 0.29 and Kendall’s $\tau$ of 0.13. This indicates a modest association, concentrated in the lower tail. In the second tree, we observe that SOFR and inflation given Treasury yields are independent. This implies that any residual dependence is negligible once Treasury yields are taken into account; see Figure \ref{fig:contour_psnorm}.

\section{Simulation analysis and contract design implications}
\label{sec:empirical}

Building on the statistical calibration developed in Section~\ref{sec:statistics1}, we
implement the pricing and contract-design framework of
Section~\ref{sec:framework} via simulation; see Section \ref{sec:na_chain} as well. Crypto-native loss dynamics are
modeled on $(\Omega^{L},\mathcal{F}^{L},\mathbb{P}^{L})$, while financial market
variables  are modeled on
$(\Omega^{F},\mathcal{F}^{F},\mathbb{P}^{F})$. The joint model is defined on the
product space
\[
(\Omega,\mathcal{F},\mathbb{P})
=
(\Omega^{L}\times\Omega^{F},\;
\mathcal{F}^{L}\otimes\mathcal{F}^{F},\;
\mathbb{P}^{L}\otimes\mathbb{P}^{F}).
\]
For a realization $(\gamma^{L},\gamma^{F})\in\Omega$, the on-chain settlement
rules in Section~\ref{sec:framework} generate the contract cashflows
\[
CF_k(\gamma^{L},\gamma^{F})
=
d\!\left(R_k(\gamma^{F}),\,Y_k(\gamma^{L}),\,Z_k(\gamma^{L})\right),
\qquad k=1,\ldots,T,
\]
where $d(\cdot)$ is defined in Eqs. \eqref{eq:cashflow}--\eqref{eq:principle},
$(Y_k,Z_k)$ denote the oracle-reported monthly maximum and cumulative loss
statistics, and $R_k$ is the monthly compounded SOFR rate. The discounted payoff
along $(\gamma^{L},\gamma^{F})$ is
\begin{equation}
\label{eq:discounted_payoff_path}
H(\gamma^{L},\gamma^{F})
=
\sum_{k=1}^{T}
\frac{
CF_k(\gamma^{L},\gamma^{F})
}{
D_k(\gamma^{F})
}.
\end{equation}
Because the trigger mechanism is nonlinear and path dependent, a closed-form characterization of
the distribution of $H$ is generally infeasible. We therefore approximate the payoff distribution
via Monte Carlo simulation. We specify the nominal discount factor along a financial
factor path $\gamma^{F}$ as
\[
D_k(\gamma^{F})
=
\prod_{s=0}^{k-1}
\bigl(1+\mathrm{TR}_s(\gamma^{F})\bigr)
\bigl(1+\mathrm{IF}_s(\gamma^{F})\bigr),
\]
where $\mathrm{TR}_s$ and $\mathrm{IF}_s$ denote the Treasury rate and inflation rate in period $s$,
respectively. The simulated distribution is then used (i) to compute par spreads by matching expected discounted
payoffs under the pricing measure $\mathbb{Q}$, where $D_k(\gamma^{F})$ provides the corresponding
discount along the simulated financial path, and (ii) to assess downside tail exposure
under the physical measure $\mathbb{P}$. In implementation, and given the one-year horizon, we approximate the financial pricing measure
for the auxiliary macro factors driving $\mathrm{TR}_s$ and $\mathrm{IF}_s$ by their historically
estimated dynamics. Equivalently, we assume that the associated interest-rate risk premia
 are negligible over this horizon. Valuation is then carried out by discounted-cashflow
simulation using the corresponding money-market numeraire, i.e., discounting by
$D_k(\gamma^{F})$ along each simulated financial path $\gamma^{F}$.

We study contract valuation and design
implications by comparing simulated price,  and tail-risk characteristics
across alternative trigger grids and principal repayment schedules. All results
reported in this section are obtained from the joint loss--financial model
estimated in Section~\ref{sec:statistics1}. We focus on the ETH category as a representative case; corresponding results for
BSC and Other protocols are reported in  the Supplementary Material.
Using the calibrated model, we generate the predictive distribution of
one-year discounted payoffs and holding-period returns for each candidate
design, and we summarize the implied pricing and sponsor risk trade-offs under
the dual-measure framework.

\paragraph{Trigger structure and principal designs.}
To determine the joint trigger state \(A_\xi B_\eta\), we partition the space
of realized losses into discrete intervals. Specifically, the monthly maximum
loss is divided into four intervals:
\[
A_1 = [0, \delta_1^M), \quad
A_2 = [\delta_1^M, \delta_2^M), \quad
A_3 = [\delta_2^M, \delta_3^M), \quad
A_4 = [\delta_3^M, \infty),
\]
and the cumulative monthly loss is divided into three bins:
\[
B_1 = [0, \delta_1^C), \quad
B_2 = [\delta_1^C, \delta_2^C), \quad
B_3 = [\delta_2^C, \infty).
\]

The joint trigger state \(A_\xi B_\eta\) determines coupon payment multipliers,
which are identical across all designs and follow the grid shown in
Figure~\ref{fig:coupon}. Principal repayment at maturity depends solely on the
cumulative loss bin \(B_\eta\) and varies across contract designs.

\begin{table}[htb!]
\centering
\caption{Principal repayment multipliers \(\Gamma(B_\eta)\) defining the four
representative contract designs \(L_1\)--\(L_4\). Designs are ordered from
most aggressive to most conservative in terms of principal protection.}
\label{tab:principal_grids_L1_L4}
\begin{tabular}{lccc}
\toprule
Design & \(\Gamma(B_1)\) & \(\Gamma(B_2)\) & \(\Gamma(B_3)\) \\
\midrule
\(L_1\) (Aggressive)    & 1.00 & 0.60 & 0.00 \\
\(L_2\) (Intermediate)  & 1.00 & 0.80 & 0.20 \\
\(L_3\) (Moderate)      & 1.00 & 0.90 & 0.80 \\
\(L_4\) (Conservative)  & 1.00 & 0.95 & 0.90 \\
\bottomrule
\end{tabular}
\end{table}
Consistent with the existence result of Theorem~3.3, we consider a finite
admissible subset \(D_N \subset D\) consisting of four representative
principal repayment designs \(L_1\)--\(L_4\), summarized in
Table~\ref{tab:principal_grids_L1_L4}. Design \(L_1\) corresponds to an
aggressive cliff-type structure, while designs \(L_2\) and \(L_3\) introduce
progressively smoother loss sharing. Design \(L_4\) represents a conservative
structure that preserves most principal even under severe cumulative losses.

\begin{algorithm}[htb!]
\small
	\caption{Simulation of predictive distribution of crypto CAT bond prices under joint crypto-financial model.}
	\label{alg:detailed_catbond_simulation}
	\SetKwInOut{Input}{\textbf{Input}}\SetKwInOut{Output}{\textbf{Output}}
	
	\Input{
	Number of simulations \( S \); Maturity \( T \);  Marginal models \(F_{k}, G_{k}\) for monthly maximum and aggregate loss; Fitted copula structures \(C_1\), $C_2$;  ARIMA-GARCH models for  $(R_k,\text{TR}_k,\text{IF}_k)$; Face value $F$.
	}
	
	\Output{Simulated predictive distribution \(\{P^{(j)}\}_{j=1}^{S}\) of the CAT bond price.}
	\BlankLine

	\For{\(j = 1, 2, \dots, S\)}{
	     Simulate  random vectors \(\{(U_k^{(j)}, V_k^{(j)})\}_{k=1}^T\) from the bivariate copula $C_1$. \\
	    Obtain simulated monthly maximum losses and monthly aggregate losses:
	    \[
	    Y_k^{(j)} = F_{k}^{-1}(U_k^{(j)}), \quad X_k^{(j)} = G_{k}^{-1}(V_k^{(j)}), \quad k=1,\dots,T.
	    \]\\
        Compute cumulative losses $Z_k^{(j)}=\sum_{i=1}^k X^{(j)}_i$, $k=1,\dots,T$.
        \BlankLine
	     Simulate   random vectors \(\{(Z_{1,k}^{(j)}, Z_{2,k}^{(j)}, Z_{3,k}^{(j)})\}_{k=1}^{T}\) from  copula \(C_2\). \\
	     Obtain simulated financial innovations via the inverse of the fitted SGED for Treasury, fitted SGED for inflation, and SNORM for SOFR.\\
         Obtain simulated financial risks via corresponding ARIMA-GARCH models with the simulated financial innovations.\\
        \BlankLine
	    \For{\(k = 1,2,\dots,T\)}{
	         Determine the trigger state by identifying the bin into which the simulated \((Y_k^{(j)}, Z_k^{(j)})\) falls, i.e.,
	        \[
	        (Y_k^{(j)}, Z_k^{(j)}) \in A_\xi B_\eta,
	        \]
	        and retrieve the corresponding coupon multiplier \(\Delta^{(j)}\) and principal multiplier \(\Gamma^{(j)}\). \\
	Evaluate the nominal cash flow function:
\[
d_k^{(j)} \;=\; d\!\left(R_k^{(j)}, Y_k^{(j)}, Z_k^{(j)}\right) 
= \left[\Delta^{(j)} \cdot \big(R_k^{(j)} + s\big) 
+ \mathbf{1}_{\{k=T\}} \cdot \Gamma^{(j)}\right]F .
\]
        \BlankLine
        }
	    Using the simulated financial paths for \(\text{TR}_s^{(j)}\) and \(\text{IF}_s^{(j)}\), calculate the present value:
	    \[
	    P^{(j)} = \sum_{k=1}^T \frac{d_k^{(j)}}{\prod_{s=0}^{k-1}\left(1+\text{TR}_s^{(j)}\right)\left(1+\text{IF}_s^{(j)}\right)}.
	    \]

    } 
	\textbf{Return:} Simulated predictive distribution \(\{P^{(j)}\}_{j=1}^{S}\).
\end{algorithm}

To determine the trigger thresholds, we first set the spread parameter to
$s=0$ in Eq.~\eqref{eq:coupon-with-spread} and search over candidate grids for
$\{\delta_1^M,\delta_2^M,\delta_3^M,\delta_1^C,\delta_2^C\}$. The trigger levels
are chosen so that, under the baseline spread-free contract, the simulated mean
annualized return is approximately $15\%$. We then use
Algorithm~\ref{alg:detailed_catbond_simulation} to approximate the distribution
of CAT bond prices via Monte Carlo simulation. In particular, using the
Treasury rate $r_0=4.959\%$ and inflation rate $\pi_0=3.4\%$ (as of
December~2023) as discounting inputs, we generate $100{,}000$ simulated
per-note prices for the ETH category under four trigger mechanisms,
$L_1$--$L_4$, with maturity fixed at $T=12$ months and face value set to
$F=1{,}000$. Table~\ref{tab:trigger_levels_eth} reports the resulting
prespecified trigger levels for ETH; corresponding results for the BSC and
Other categories are provided in the Supplementary Material.


\begin{table}[htb!]
    \centering
    \caption{Prespecified trigger levels for ETH. 
    Results for BSC and Other categories are reported in the Supplementary Material.}
    \label{tab:trigger_levels_eth}
    \begin{tabular}{lccccc}
        \toprule
        Level & $\delta_1^M$ & $\delta_2^M$ & $\delta_3^M$ & $\delta_1^C$ & $\delta_2^C$ \\
        \midrule
        $L_1$–$L_4$ & $1.99\times10^{6}$ & $6.68\times10^{6}$ & $1.18\times10^{7}$ & $2.37\times10^{9}$ & $2.56\times10^{9}$ \\
        \bottomrule
    \end{tabular}
\end{table}

Table~\ref{tab:trigger_levels_eth} summarizes the trigger thresholds that define the loss partitions in the contract structure. The monthly-maximum trigger is characterized by three levels, with $\delta_1^M = 1.99\times10^{6}$, $\delta_2^M = 6.68\times10^{6}$, and $\delta_3^M = 1.18\times10^{7}$. 
For the cumulative trigger, the two thresholds $\delta_1^C = 2.37\times10^{9}$ and $\delta_2^C = 2.56\times10^{9}$ define the transition into cumulative-loss states.

\paragraph{Optimal design and tail risk.}
We illustrate sponsor-side contract design  by
conducting a discrete comparison across a finite set of smart-contract-feasible
principal schedules $\Gamma$ (levels $L_1$--$L_4$) and evaluating the resulting
sponsor tail exposure. For each candidate design, we first calibrate the coupon
spread $s$ to satisfy the par-issuance condition in Eq. \eqref{eq:par_tail}.
Specifically,   we estimate the mean discounted payoff   implied by the
design and solve for the spread $s$ such that the simulated bond price equals
the face value $F$; see also Eq. \eqref{eq:spead}. The resulting calibrated spreads are reported in
Table~\ref{tab:spread}.

\begin{table}[htb!]
\centering
\caption{Spread value $s$ under different trigger levels.}
\label{tab:spread}
\footnotesize
\begin{tabular}{lcccc}
\toprule
Trigger level & $L_1$ & $L_2$ & $L_3$ & $L_4$ \\
\midrule
Spread value $s$ & 1.1809\% & 1.1451\% & 1.0823\% & 1.0689\% \\
\bottomrule
\end{tabular}
\end{table}


We then assess sponsor risk under the physical measure $\mathbb{P}$ by
computing $\mathrm{TVaR}_\alpha$ of the sponsor shortfall
$X_N=(L-P_N(\Delta,\Gamma))_{+}$ for $\alpha\in\{0.95,0.99\}$, where the total
collateralized principal is fixed at $NF=300$M. This corresponds to the
sponsor-side objective in Eq. \eqref{eq:tail_design_problem}: among par-calibrated
contracts,
the sponsor prefers designs that minimize residual tail exposure under
$\mathbb{P}$. Table~\ref{tab:tvar} shows that, at both confidence levels $\alpha=0.95$ and $\alpha=0.99$, the lowest TVaR values are attained under trigger level $L_1$. This outcome reflects the fact that $L_1$ corresponds to the most aggressive principal multiplier structure. Under this design, investor principal is written down more rapidly once losses enter adverse regions, so that a larger portion of extreme losses is absorbed by the bond rather than retained by the sponsor. Consequently, in scenarios with very large losses, the sponsor benefits from stronger protection, which translates into a smaller expected tail shortfall. 
\begin{table}[htb!]
\centering
\caption{Sponsor shortfall TVaR under $\mathbb{P}$ across trigger levels with total principal $NF=300$M.}
\label{tab:tvar}
\begin{tabular}{lcccc}
\toprule
Trigger level & $L_1$ & $L_2$ & $L_3$ & $L_4$ \\
\midrule
TVaR$_{95}$ & $2.301\times10^{9}$ & $2.325\times10^{9}$ & $2.367\times10^{9}$ & $2.376\times10^{9}$ \\
TVaR$_{99}$ & $2.527\times10^{9}$ & $2.587\times10^{9}$ & $2.739\times10^{9}$ & $2.768\times10^{9}$ \\
\bottomrule
\end{tabular}
\end{table}
It is important to note that the trigger thresholds are calibrated so that the investor’s expected return is
approximately $15\%$ under all four principal-multiplier levels. From the investor’s perspective, the four levels
therefore deliver broadly comparable expected returns, whereas from the sponsor’s perspective a more aggressive
level can reduce expected tail risk through stronger protection. We emphasize that, although
Theorem~\ref{thm:tail_exist} establishes existence of an optimal contract design in the continuous formulation, we do
not attempt to solve the full optimization problem in this empirical section. Instead, we conduct a 
 grid-based search over a  set of candidate trigger thresholds and multipliers to illustrate the
risk-return trade-offs and the qualitative implications of the theoretical results.



\section{Case Study: Design and Feasibility of a Chain-Wide ETH Catastrophe Bond}
\label{sec:case-study}

To illustrate the economic viability and operational feasibility of the proposed
on-chain settlement framework, we present a case study of a one-year, fully
collateralized ETH catastrophe bond. The instrument is designed to transfer
chain-wide crypto-system risk from a sponsor to capital market investors via the proposed framework.

In traditional markets, such SPVs (e.g., Long Walk Re\footnote{\url{https://www.artemis.bm/deal-directory/long-walk-re/} — Bermuda-based special purpose insurer sponsored by AXIS Capital, used for cyber and property catastrophe bond issuances since 2023.}), which act as dedicated issuing vehicles for specific risk transfer programs. The sponsor (i.e., the protection buyer seeking coverage against
blockchain-related catastrophic loss) is analogous to the cedant in the
traditional catastrophe bond market, such as (re)insurers and reinsurers
(e.g., AXIS Capital or Munich Re). It may also be a crypto-native risk carrier
or mutual-style protection provider (e.g., \emph{Nexus Mutual}\footnote{\url{https://nexusmutual.io/} --- an Ethereum-based mutual that offers cover for smart-contract and related crypto risks.}), or a large centralized crypto intermediary (e.g., major exchanges/custodians such as Coinbase or Binance) that seeks to hedge severe operational loss scenarios.
The legal structure mirrors market convention: SPV issuance and trustee-held collateral, while the blockchain layer is used  for calculation and distribution. Each month, the index engine computes the state $(Y_k,Z_k)$ and posts a signed record on-chain, which the contracts translate into the coupon multiplier; at maturity a final state produces the principal multiplier.  The rationale for the choice of this hybrid design, retaining the traditional SPV issuance and trustee custody structure while adding a blockchain-based operational layer, is as follows:
From a legal and regulatory standpoint, the SPV framework is a proven vehicle for isolating risk and securing collateral, familiar to rating agencies, regulators, and institutional ILS investors. This ensures that the product can be evaluated, documented, and distributed within existing market infrastructure without altering the legal underpinnings that investors rely on. At the same time, the blockchain layer delivers capabilities that are difficult to replicate off-chain: automated coupon and principal calculation, transparent publication of loss metrics, and verifiable distribution records. By limiting on-chain functions to calculation and settlement, the design avoids extra custody risk and structural frictions for investors. The underlying pricing and risk transfer remain governed by our model outputs. For example, consider an ETH catastrophe note under the $L_1$ trigger schedule. If at maturity the cumulative loss $Z_T$ falls into bin $B_2$, the pricing model specifies a principal multiplier 
$\Gamma=0.6$. Within the on-chain system, the Trigger Contract enforces this rule by reducing the redeemable 
principal to $60\%$ of par and updating investor token balances accordingly. 
This provides a direct and transparent implementation of the pre-calibrated payout logic, with no scope for ex post discretion. Similarly, each month the realized state $(Y_k,Z_k)$ is mapped to a cell $A_\xi B_\eta$, 
which determines the coupon multiplier $\Delta$. 
These multipliers are passed to the Coupon Distribution Contract, which computes the adjusted SOFR-linked 
coupon payments and makes them available to investors on a pro rata basis.

Since our simulation results indicate that, within the candidate family
$\mathcal{D}_N=\{L_1,\ldots,L_4\}$, the aggressive principal schedule $L_1$
achieves the strongest transfer of extreme tail exposure, we adopt $L_1$ as the benchmark
design for illustrating investor return outcomes. For each simulated scenario $j=1,\ldots,100{,}000$, we compute the one-year
holding-period return implied by the contract cash flows. Let $d_k^{(j)}$
denote the nominal cash flow paid at period $k$ in simulation $j$.
The return rate for simulation $j$ is defined as
\[
\mathrm{RR}^{(j)}=\frac{\sum_{k=1}^{T} d_k^{(j)}}{F}-1.
\]
This yields a sample of $100{,}000$ simulated return rates, each corresponding
to a distinct joint realization of crypto-loss dynamics and financial market
conditions. In addition, we report an illustrative \emph{pure realized return benchmark}
based on observed market outcomes in 2024. Let   $\mathrm{MV}$ denote the realized total value received at maturity
(including coupons and principal repayment) under the $L_1$ settlement rule
applied to the 2024 realized path. The pure realized return is then
\[
\mathrm{RR}_{\mathrm{real}}=\frac{\mathrm{MV}}{F}-1.
\]

\begin{table}[htb!]
\centering
\caption{Summary of simulated and pure realized return rates for the ETH category under trigger design $L_1$.}
\label{tab:eth_return_l1}
{
\begin{tabular}{llrrrrrrrr}
\toprule
\multicolumn{10}{c}{\textbf{Panel A: Simulated return rates}} \\
\midrule
Type & Trigger & 2.5\% & Q$_1$ & Median & Mean & Q$_3$ & 97.5\% & SD \\
\midrule
Simulated & $L_1$ & 0.13 & 0.15 & 0.16 & 0.16 & 0.17 & 0.19 & 0.01\\
\midrule
\multicolumn{10}{c}{\textbf{Panel B: Pure realized return benchmark (2024)}} \\
\midrule
\multicolumn{2}{l}{Purchase price} & \multicolumn{8}{c}{$L_1$} \\
\cmidrule(lr){1-2}\cmidrule(lr){3-10}
\multicolumn{2}{l}{Par ($F=1000$)} & \multicolumn{8}{c}{18.9\%} \\
\bottomrule
\end{tabular}}
\end{table}


Table~\ref{tab:eth_return_l1} reports the distribution of simulated return rates for the ETH category under trigger design $L_1$. Panel A indicates that both the median and the mean return are 16\%, suggesting a fairly symmetric central tendency. The central 95\% interval of simulated returns ranges from 13\% to 19\%, and the standard deviation is about 1\%, reflecting a relatively concentrated return distribution under this trigger structure. Panel B provides a pure realized return benchmark based on the par reference value ($F=1000$). Under the observed path, the realized return equals 18.9\%, which lies toward the upper region of the simulated distribution.

\paragraph{Sponsor cost components and evaluation.} A key contribution of the proposed on-chain settlement mechanism is the modification of the transaction cost curve. From the sponsor’s perspective, the total cost of risk transfer under the $L_1$ design consists of three main components:

\begin{enumerate}
    \item \textbf{Coupon payments.} These constitute the primary, recurring cost. Under the 2024 realized path, the pure realized holding-period return for $L_1$
is $18.9\%$. {Since principal is fully repaid} and the
purchase price is par, the realized return equals the
coupon fraction, so the sponsor's realized total coupon outlay is
$\$56.7$ million for $NF=\$300$ million.

    \item \textbf{On-chain administration.} The blockchain settlement layer, described in Section~2, executes monthly coupon state updates and a final principal settlement. Implemented on an  Ethereum-compatible Layer-2  network, such as Arbitrum or Optimism, with 2024 gas fees (e.g., 300{,}000 gas units at 0.1 gwei, ETH/USD in the low thousands), the annualized cost remains well below  {\$100} \citep{L2Fees2025}. We use an on-chain cost estimate of \$14{,}000, which reflects a conservative upper bound.\footnote{At prevailing 2024 Layer-2 gas prices, monthly state updates and final settlement require less than \$100 in total \citep{L2Fees2025}. The dominant component is an annual oracle subscription (typically around \$10{,}000 for a dedicated feed), with the remainder allocated to redundancy, monitoring, and contingency.} This cost is operationally negligible and has no material impact on the sponsor’s financial calculus.

    \item \textbf{Off-chain fixed cost.} As in conventional ILS programs, a one-year issuance entails fixed administrative and legal outlays: SPV formation and maintenance, trustee and custody, offering documentation, regulatory filings, and third-party modeling / rating \citep{NAICILS2020}. Consistent with market practice of parameterizing issuance costs either as a fixed dollar amount or as a percent of limit \citep{HKLegCoFSTB2021, GallagherRe2025}, we adopt two stylized specifications for off-chain issuance/administration in our scenarios: (i) a fixed benchmark of \$0.75 million per one-year deal and (ii) a sensitivity band of 1–2\% of notional\footnote{Although the proposed on-chain settlement layer can in principle reduce recurring costs 
by eliminating calculation and paying-agent fees, we do not reflect these savings in our base scenarios. 
First-year issuance will still incur the full legal, trustee, and regulatory budget, 
and published market benchmarks for issuance costs are reported on a total basis (fixed dollar 
or percent of notional). For comparability with conventional ILS studies, we therefore retain 
the 1--2\% specification here, while noting that steady-state recurring costs would plausibly 
be lower in practice once contracts and infrastructure are reused.}.

\end{enumerate}

\begin{table}[htb!]
\centering
\caption{Sponsor cost components with explicit on-chain fees (ETH, $L_1$; $NF=\$300$M; pure realized coupon-equivalent return $=18.9\%$).}
\label{tab:onchain_recalc}
\footnotesize
\begin{tabular}{lrrrrr}
\toprule
\textbf{Admin specification} & \textbf{Coupon} & \textbf{Off-chain admin} & \textbf{On-chain} & \textbf{Total (USD)} & \textbf{\% of limit}\\
\midrule
Fixed admin \$0.75M      & \$56{,}700{,}000 & \$750{,}000   & \$14{,}000 & \$57{,}464{,}000 & 19.15\% \\
Admin = 1\% of notional  & \$56{,}700{,}000 & \$3{,}000{,}000 & \$14{,}000 & \$59{,}714{,}000 & 19.90\% \\
Admin = 2\% of notional  & \$56{,}700{,}000 & \$6{,}000{,}000 & \$14{,}000 & \$62{,}714{,}000 & 20.90\% \\
\bottomrule
\end{tabular}
\end{table}

\noindent Table~\ref{tab:onchain_recalc} summarizes the sponsor’s cost structure
for an ETH-denominated catastrophe note under the $L_1$ trigger calibration with
total principal fixed at $NF=\$300$M. The coupon component is reported using the
 {pure realized} 2024 coupon-equivalent return of $18.9\%$ under $L_1$. We have the following observations.

\begin{itemize}
    \item \textbf{Coupon dominates total cost.}
    Coupon payments are the main cost driver, accounting for the overwhelming
    majority of total sponsor outlay across all scenarios (e.g., \$56.7M out of
    \$57.464M under the fixed-cost benchmark). By contrast, the on-chain
    component is operationally negligible (\$14{,}000, i.e., less than
    $0.01\%$ of limit), consistent with the view that the blockchain layer
    primarily changes the \emph{execution} cost curve rather than the economic
    cost of risk transfer.

    \item \textbf{Off-chain administration affects the effective cost ratio.}
    Varying off-chain issuance/administration assumptions from a fixed \$0.75M
    to 1--2\% of notional changes the all-in sponsor cost from 19.15\% to 20.90\%
    of limit. Although this variation is small relative to the coupon component,
    it is economically meaningful at scale: the total cost increases by about
    \$5.25M when moving from the fixed-cost benchmark to the 2\% specification
    (i.e., \$57.464M vs.\ \$62.714M for $NF=\$300$M). This highlights the
    importance of how off-chain costs are benchmarked in practice
    \citep{NAICILS2020,HKLegCoFSTB2021,GallagherRe2025}.

    \item \textbf{Level relative to mainstream ILS and calibration choice.}
    The all-in cost level implied by the 2024 realized benchmark (roughly
    19--21\% of limit) remains toward the high end relative to mature
    natural-cat ILS, reflecting both emerging-peril compensation
    \citep{GenevaAssociation2024CyberILS,ArtemisBeazley2024} and our calibration
    choices. In particular, the prespecified thresholds in
    Table~\ref{tab:trigger_levels_eth} were selected to target a relatively high
    expected investor return in the simulation design stage, which mechanically
    translates into a higher coupon-equivalent sponsor cost when principal is
    fully repaid. In practice, sponsors can reduce the required cost level by
    adjusting trigger sensitivity and payout schedules (e.g., shifting
    thresholds upward to reduce expected trigger incidence, or increasing
    $\Gamma$ in adverse bins to reduce expected compensation), thereby targeting
    return levels closer to prevailing ranges in cyber/ILS markets.\footnote{This
    reflects a standard design trade-off: lowering the investor return
    requirement reduces sponsor cost but also reduces the extent of tail-risk
    transfer.}
\end{itemize}

\section{Conclusion and Discussion} \label{sec:conclusion}

This paper develops a   multi-trigger framework for crypto
catastrophe instruments with on-chain settlement. On the theoretical side, we
(i) formalize a trigger-driven cash-flow architecture based on oracle-reported
monthly maximum and cumulative loss statistics, (ii) establish an arbitrage-free
valuation framework under an enlarged information filtration and a pricing
measure $\mathbb{Q}$, and (iii) derive a martingale-based decomposition that
separates the financial discounting component from the crypto-trigger component. We then
formulate sponsor-optimal contract design under a dual-measure setting in which
market valuation is governed by $\mathbb{Q}$ while sponsor exposure is assessed
under the physical measure $\mathbb{P}$, and we show that an optimal design
exists within the admissible smart-contract-feasible design class. Empirically, we operationalize the theory through a joint crypto--financial
statistical model. On the crypto side, monthly maximum and cumulative losses are
modeled using flexible GEV-based marginals with dependence captured by copulas;
on the financial side, Treasury rates, inflation, and SOFR dynamics are modeled
via ARIMA--GARCH marginals and a vine copula for cross-series dependence. These
models generate joint scenario paths that feed directly into the on-chain payout
maps, yielding predictive distributions for discounted bond prices, holding-period
returns, and sponsor shortfall via Monte Carlo simulation. A 2024 replay provides
an implementation-level check: applying the settlement rules to the realized path
produces outcomes consistent with a low-attachment year and lying within the
model-implied predictive envelope.

This work also has several limitations, similar to other studies in the ILS literature. 
First, non-stationarity and structural breaks in DeFi risk and macroeconomic conditions may challenge the validity of static copula assumptions. However, with the limited monthly data, dynamic dependence models are difficult to estimate reliably in the present setting. Static copulas therefore provide a reasonable first step, while longer time series will enable more robust exploration of time-varying structures in future work. Second, parameter uncertainty is not fully propagated into pricing, which can be addressed via the Bayesian approach. However,  given the limited sample sizes available here, Bayesian posterior simulation would add considerable complexity without necessarily improving robustness. We view full uncertainty quantification as a valuable direction for future work, particularly to produce credible intervals for bond prices and tail risk measures.  Third, beyond the
hybrid structure studied here, an important research
direction is the design of fully tokenized catastrophe instruments, issued,
held, and potentially traded natively as blockchain tokens, which could enable
more autonomous risk transfer and improved programmability in secondary-market
trading. Realizing this vision will require addressing two practical obstacles:
(i) regulatory and legal constraints on tokenized securities and cross-border
distribution, and (ii) collateral and valuation risks induced by crypto-market
volatility, which can affect both the stability of posted collateral and the
reliability of token-based settlement.

\medskip
{\color{blue}
\paragraph{Acknowledgments}
{The authors are grateful for the referee's constructive comments, which led to an improved version of the article.}
}

\bibliographystyle{apalike-ejor}
\bibliography{cat-crypto,att-all,breach,crypto,refbreach}
\end{document}